\documentclass[11pt]{article} % use larger type; default would be 10pt

\usepackage{amsfonts}
\usepackage{amssymb,amsmath}
\usepackage{amsthm,amsgen}
\usepackage{graphicx,epsfig}
\usepackage[mathscr]{eucal}
\usepackage{bm}
\usepackage{bbm}
\usepackage{ marvosym }
\usepackage{dsfont}
\usepackage{accents}
\usepackage{hyperref}
\usepackage{rawfonts}
\usepackage{pictex}
\usepackage{color}

\newtheorem{thm}{THEOREM}[section]
\newtheorem*{thm*}{THEOREM}
\newtheorem{cor}[thm]{Corollary}

\newtheorem{prop}[thm]{Proposition}
\theoremstyle{definition}

\newtheorem{rem}[thm]{Remark}

\DeclareMathOperator{\tr}{tr}

\newcommand{\refeq}[1]{{\rm(\ref{#1})}}
\renewcommand{\eqref}[1]{{\rm(\ref{#1})}}

\newcommand{\bpi}{\boldsymbol{\pi}}

\numberwithin{equation}{section}

\newcounter{subequation}
	\newenvironment{subequation}%
	{\addtocounter{equation}{-1}%
	\stepcounter{subequation}%
	\begin{equation}}%
	{\end{equation}%
}

\newcommand{\be}{\mathbf{e}}

\newcommand{\bn}{\mathbf{n}}

\newcommand{\bP}{\mathbf{P}}

\newcommand{\bq}{\mathbf{q}}

\newcommand{\bT}{\mathbf{T}}

\newcommand{\bx}{\mathbf{x}}

\newcommand{\bz}{\mathbf{z}}
\newcommand{\bZ}{\mathbf{Z}}

\newcommand{\beq}{\begin{equation}}
\newcommand{\eeq}{\end{equation}}
\newcommand{\bseq}{\begin{subequation}}
\newcommand{\eseq}{\end{subequation}}

\newcommand{\p}{\partial}

\newcommand{\cA}{\mathcal{A}}
\newcommand{\cB}{\mathcal{B}}
\newcommand{\cC}{\mathcal{C}}

\newcommand{\cI}{\mathcal{I}}

\newcommand{\cM}{{\mathcal M}}
\newcommand{\cN}{{\mathcal N}}

\newcommand{\cQ}{{\mathcal Q}}

\newcommand{\bLa}{\boldsymbol{\Lambda}}

\newcommand{\psiPH}{{\psi_{\mbox{\rm\tiny ph}}}}

\newcommand{\psiEL}{{\psi_{\mbox{\rm\tiny el}}}}

\newcommand{\phiPH}{{\phi_{\mbox{\rm\tiny ph}}}}

\newcommand{\mEL}{{m_{\mbox{\rm\tiny el}}}}

\newcommand{\mPH}{{m_{\mbox{\Lightning}}}}

\newcommand{\opsi}{\accentset{\circ}{\psi}}

\newcommand{\Bset}{{\mathbb B}}
\newcommand{\Cset}{{\mathbb C}}

\newcommand{\Pset}{\mathbb{P}}

\newcommand{\Rset}{{\mathbb R}}

\newcommand{\ze}{\zeta}

\newcommand{\de}{\delta}

\newcommand{\al}{\alpha}

\newcommand{\bet}{\beta}
\newcommand{\ga}{\gamma}

\newcommand{\om}{\omega}
\newcommand{\Om}{\Omega}
\newcommand{\si}{\sigma}
\newcommand{\Si}{\Sigma}

\newcommand{\half}{\frac{1}{2}}

\newcommand{\diag}{\mbox{diag}}
\newcommand{\Span}{\mbox{Span}}

\newcommand{\bna}{\begin{eqnarray}}
\newcommand{\ena}{\end{eqnarray}}
\newcommand{\bea}{\begin{eqnarray*}}
\newcommand{\eea}{\end{eqnarray*}}
\newcommand{\ben}{\begin{enumerate}}
\newcommand{\een}{\end{enumerate}}
\newcommand{\bi}{\begin{itemize}}
\newcommand{\ei}{\end{itemize}}

\newcommand{\dal}{\raisebox{3pt}{\fbox{}}\,}

\newcommand{\el}{{\mbox{\tiny{el}}}}
\newcommand{\ph}{{\mbox{\tiny{ph}}}}

\title{A Lorentz-Covariant Interacting Electron-Photon System \\ in One Space Dimension}
\author{\normalsize\sc{Michael K.-H. Kiessling$^{*,\diamond}$, Matthias Lienert$^\dagger$, and A. Shadi Tahvildar-Zadeh$^{*}$}\\
	{$\phantom{nix}$}\\[-0.1cm] 
        \normalsize $^{*}\,$Department of Mathematics\\[-0.1cm]
	Rutgers, The State University of New Jersey\\[-0.1cm]
	110 Frelinghuysen Rd., Piscataway, NJ 08854, USA\\[0.2cm]
	$^\dagger \,$Fachbereich Mathematik\\[-0.1cm]
	Eberhard-Karls-Universit\"at T\"ubingen\\[-0.1cm]
	Auf der Morgenstelle 10, 72076 T\"ubingen, Germany}
\date{Accepted for publication in Lett. Math. Phys., Sept. 18, 2020} 
\begin{document}
\maketitle
\vspace{-1truecm}
\begin{abstract} 
A Lorentz-covariant system of wave equations is formulated for a quantum-mechanical two-body system in one space dimension, 
comprised of one electron and one photon. 
 {Manifest Lorentz covariance is achieved using Dirac's formalism of multi-time wave functions, i.e., 
wave functions $\Psi^{{(2)}}(\bx_\ph,\bx_\el)$ where $\bx_\el,\bx_\ph$ are the generic spacetime 
events of the electron and photon, respectively.} 
Their interaction is implemented via a Lorentz-invariant no-crossing-of-paths boundary condition at the coincidence submanifold 
$\{\bx_\el=\bx_\ph\}$, compatible with particle current conservation.
 The corresponding initial-boundary-value problem is proved to be well-posed.
 Electron and photon trajectories are shown to exist globally in a hypersurface {Bohm--Dirac} theory, for typical particle
initial conditions. 
 Also presented are the results of some numerical experiments which illustrate Compton scattering as well as a new phenomenon:
photon capture and release by the electron.\vspace{-10pt}
\end{abstract}
\vfill
\hrule
\vspace{5pt}
\copyright{2020} The authors. Reproduction of this preprint is permitted for non-commercial purposes.\hspace{1truecm}
\hrule
$^\diamond$ Corresponding author. email: miki@math.rutgers.edu

\section{Introduction and summary of main results}

\begin{quote}\small\it
``The Compton effect, at its discovery, was regarded as a simple collision of two bodies, and yet the 
detailed discussion at the present time involves the idea of the annihilation of one photon and the 
simultaneous creation of one among an infinity of other possible ones. We would like to be able to 
treat the effect as a two-body problem, with the scattered photon regarded as the same individual as 
the incident, in just the way we treat  the collisions of electrons.''

\hfill ---{\bf C G. Darwin,} ``Notes on the Theory of Radiation" (1932)
\end{quote}\vspace{-.3truecm}
 Four score and {eight years after Charles G.~Darwin's paper appeared in print, it is still true that 
``the detailed discussion [of the Compton effect] at the present time involves the idea of the 
annihilation of one photon and the simultaneous creation of one among an infinity of other possible ones;''
 see \cite{FGS,CFPa,CFPb,Buchholz,AlDy} for rigorous quantum-field theoretical treatments.
 Not only has Darwin's goal, of treating electrons and photons on an equal footing 
within a quantum-{\em mechanical} framework of a fixed number of particles, remained elusive, 
during the last century quantum-field  theorists have come to the conclusion that 
``a relativistic quantum theory of a fixed number of particles, is an impossibility;'' (\cite{WeinbergBOOKqft}, p.3).}

{The purpose of the present paper is to demonstrate that the impossible can be done
for a 1+1-dimensional toy version of the problem. 
 More precisely, we show how the relativistic quantum-mechanical photon wave equation of \cite{KTZ2018} can be coupled with 
Dirac's well-known relativistic quantum-mechanical} wave equation for the electron in a Lorentz-covariant manner to accomplish 
Darwin's goal:
 to ``treat the [Compton] effect as a two-body problem, with the scattered photon regarded as the same individual as the incident,
 in just the way we treat the collisions of electrons.'' 

 True, {a Lorentz-covariant dynamical theory of one electron and one photon in 1+1-dimensional Minkowski spacetime is 
merely a caricature of a relativistic quantum-mechanical theory of an arbitrary but fixed finite number $N$ of electrons, 
photons, and their anti-particles in 3+1-dimensional Minkowski spacetime, yet also quantum field theory (QFT) once started 
in 1+1 dimensions.
 And while quantum electrodynamics (QED) in particular is 
often lauded as one of the most successful physical theories ever developed by physicists, this praise often ignores that QED
is not a mathematically consistent physical theory in the same sense in which, say, Newton's mechanics is (for instance),
and neither is the QFT known as the standard model of elementary particles; cf. \cite{Penrose}, \cite{Jost}, and \cite{Bricmont}
for similar sentiments.
 In this situation one should not only continue to try to overcome QFT's technical and conceptual problems, as e.g. in \cite{Buchholz}, 
one should also explore alternative approaches to relativistic quantum theory, in particular when they have been declared to be impossible.}

  {In this vein, we emphasize that the relativistic wave equation for the
quantum-mechanical photon wave function constructed in \cite{KTZ2018} furnishes a conserved probability current 
for the photon position, obeying Born's rule, which transforms in the right way under Lorentz transformations.
 Such a feat has often been declared to be impossible, too; see, for instance \cite{BohmBOOKonQM}.}
 
 {Mathematically well-formulating a ``physical theory'' is not in itself sufficient to establish 
a physical theory of natural phenomena.
 Without quantitative explanatory power of empirical data no ``physical theory'' would be acceptable in this sense.
 Of course, a 1+1-dimensional treatment of the Compton effect is too simplistic to allow a comparison with empirical data. 
{Moreover,} a comparison of our 1+1-dimensional relativistic toy electron-photon quantum mechanics 
with a 1+1-dimensional toy QED is not possible {either}, since QED is based on the quantization of the Maxwell--Dirac system, 
{and} 
Maxwell{'s} equations in 1+1 dimensions do not describe electromagnetic radiation 
{that} could be quantized.
 Our present work should be seen as a \emph{proof of concept} only.}
 
 {The just{-}mentioned fact that in 1+1 spacetime dimensions there are no 
electromagnetic Maxwell radiation fields makes it plain that the photon wave function 
constructed in \cite{KTZ2018} for arbitrary $1+d$-dimensional Minkowski spacetimes is not
a Maxwell field, and the photon wave equation in \cite{KTZ2018} is not Maxwell{'s} field equations {in disguise}.
{What {\em is} true is that,} in $1+3$ dimensions{,} the mathematical structure of Maxwell's field equations 
{shows up} as a 
substructure in the quantum-mechanical photon wave equation of \cite{KTZ2018}, and while 
there is {therefore} some 
danger of confusing the two with each other, one should not.
 Not only is one structure mathematically a subset of the other, the Maxwell fields are defined in physical 
spacetime while the photon wave function is defined on the configuration spacetime of the photon, 
{which is a totally different arena, even though for a single photon the two spacetimes happen to be isomorphic}.}

{It follows right away from what we have just stated{,} that the interaction of a photon and an electron in a 1+1-dimensional
model cannot possibly be achieved by the familiar ``minimal coupling'' of a Dirac and a Maxwell field.
  Instead, we work with Dirac's manifestly Lorentz-covariant formalism of multi-time wave functions ---
for our $N=2$ body problem, wave functions $\Psi^{(2)}(\bx_\ph,\bx_\el)$         
{that} depend on the two generic {space-like separated}                
events $\bx_\el$ and $\bx_\ph$ of the electron and the photon, respectively 
---, and we accomplish the interaction through 
suitable {boundary conditions}
at the subset of co-incident events, $\{\bx_\el=\bx_\ph\}$, 
{which we call an \emph{intrinsic boundary}.} }

{Historically there is a direct line from Dirac's multi-time formalism to 
the quantum field-theoretical formalism of Schwinger and Tomonaga; cf. \cite{SchweberBOOK}.
 This suggests that our approach is not totally disjoint from a quantum field-theoretical formulation. 
 Indeed, the connection promises some illuminating insights into the quantum-theoretical formalism.}

{{T}he following relation connects the Heisenberg picture of quantum field theory 
with the formalism of multi-time wave functions:
 Let $\widehat{\Psi}$ stand for the Heisenberg field operators, $|\emptyset \rangle$ 
for the zero-particle state and $|\Psi_0\rangle$ for the Heisenberg state vector. 
 Then, for all spacelike configurations $(\bx_1,...,\bx_N)$ of any number $N$ of particles, the 
$N$-time $N$-body wave function is given by            
\begin{equation}\label{QFTtoMULTIt}
  \Psi^{(N)}(\bx_1,...\bx_N) = 
\frac{1}{\sqrt{N!}} \langle \emptyset| \widehat{\Psi}(\bx_1) \cdots \widehat{\Psi}(\bx_N) |\Psi_0\rangle.
\end{equation}
 Thus the multi-time picture can be regarded as a covariant particle-position representation of the quantum state.}

{Now, for 
quantum field-theoretical dynamics{,} the particle creation/annihilation formalism requires working with 
a multi-time Fock-type wave function, which can be represented as a sequence 
\begin{equation}
	\Psi = (\Psi^{(0)}, \Psi^{(1)}, \Psi^{(2)}, \cdots),
\end{equation}
where each $\Psi^{(N)}$ is an $N$-particle multi-time wave function of the kind described above.
 The quantum field-theoretical dynamics for the multi-time wave equations \eqref{abstrERWINeq}
then typically mixes the different particle numbers in the interaction terms, i.e., the right-hand side of the
abstract Schr\"odinger equation
\begin{equation}\label{abstrERWINeq}
	i\partial_{t_k} \Psi^{(N)} = (H_k\Psi)^{(N)}
\end{equation}
generally contains $\Psi^{(M)}$ for $M\neq N$.}

{The quantum-field formalism is usually applied with all times synchronized, cf. \cite{ThallerBOOK};
work on the multi-time formalism by comparison has been lagging behind. 
 Recently a multi-time formulation has been achieved for certain quantum field theories (essentially {those} 
with local interactions and finite propagation speed), see \cite{PT2014b}.
 Proofs of the existence of multi-time versions of particular quantum field theories have recently been carried out also 
when a single-time version of these quantum field theories had not been rigorously constructed before, in which case the 
relation (\ref{QFTtoMULTIt}) could not be used; see \cite{LN2020} for a toy model in 1+1 dimensions, or assuming an 
ultraviolet cutoff \cite{DN2019}.}

 {It is important to note that unique solvability of the system of evolution equations \refeq{abstrERWINeq} 
requires imposing suitable {\em boundary conditions} at the subset of co-incident events, $\{\bx_j=\bx_k\}$. 
{We refer to the set of coincidence points as an \emph{intrinsic boundary}, because it is system-intrinsic
 and not determined by some extraneous constraint (such as a container wall).
 In particular, the quantum field-theoretical creation/annihilation formalism can be implemented by formulating 
suitable intrinsic boundary conditions on the $N$-body multi-time wave functions which involve, in addition to such
boundary points, interior points of the $N-1$-body wave functions, see \cite{TT2015,TT2016} for single-time wave functions. 
 This has led to the notion of \textit{interior-boundary conditions}. For an extension to multi-time wave functions, see \cite{LN2020}.}
 In a nutshell, the difference between our quantum-mechanical approach with a fixed number $N$ of
particles and a quantum field-theoretical approach with particle creation/annihilation boils down to 
postulat{ing} 
different types of boundary conditions at the set of coincidence events.}

 {In our relativistic quantum-mechanical theory with a fixed finite number $N=2$ of particles,}
both the Dirac operator of a free electron \cite{ThallerBOOK}, and the Dirac-type operator of a 
free photon constructed in \cite{KTZ2018}, act on $\Psi^{{(2)}}(\bx_\ph,\bx_\el)$.
 Conservation of the particle current, in concert with a ``no-crossing of paths'' condition for electron and photon,
fixes the boundary condition at the subset of co-incident events, $\{\bx_\el=\bx_\ph\}$, up to a choice of a phase. 
 We will prove that the resulting initial-boundary-value problem is well-posed.

 It is intuitively obvious that a boundary condition at coincident events $\{\bx_\el=\bx_\ph\}$ 
amounts to a local pair interaction between electron and photon in a Lorentz-covariant manner. 
 Our boundary condition is compatible with the kind of interaction expected for an electron and a photon in Compton scattering. 

 Beyond demonstrating that relativistic quantum mechanics for a fixed number of two interacting particles is feasible
with a multi-time wave function, we also inquire into the possibility of computing the dynamics of these particles 
and whether their motion indeed resembles ``a simple collision of two bodies,'' as Darwin wrote.
 For this we pick up on the ideas originally conceived for photons by Einstein \cite{EinsteinPHOTONb}, \cite{EinsteinPHOTONbb}
and in 1923/24 adapted for electrons by de Broglie in his thesis. 

 Einstein thought of his ``light quanta'' (photons) as being particles which
are guided in their motion by a  ``wave field'' in spacetime{.}  
{S}ince the photons 
carried the momentum and energy, he thought of the guiding field itself as devoid of energy and momentum and thus called it a 
``ghost field,'' but he never developed his ideas into a working theory.
 Yet Einstein's ideas prompted de Broglie to argue that also electrons might be guided by a wave field, though he didn't 
supply a wave equation. 
 After Schr\"odinger's discovery  in 1926 of his wave equation, the next person to pick up on the guiding field idea was
Born  who in \cite{BornsPSISQUAREpapersA}, \cite{BornsPSISQUAREpapersB} interpreted Schr\"odinger's
wave function $\Psi$ on the configuration space of electrons and nuclei as a guiding field for these 
particles, though not stating how the guiding is being done --- except that Born thought of it as \emph{non-deterministic} 
and such that the likelihood of finding a configuration in d$^{3N}q$ centered at $Q\in \Rset^{3N}$ is $|\Psi|^2(Q){\rm{d}}^{3N}q$.
 Eventually Nelson's non-relativistic ``stochastic mechanics'' (see \cite{Guerra1995}) realized Born's narrative 
of a non-deterministic guiding law in a mathematically sharp manner.
 But also de Broglie \cite{deBroglieSOLVAY} combined the ideas of his thesis with Born's suggestion that
Schr\"odinger's $\Psi$ might guide massive particles, and came up with a deterministic theory, which was
later rediscovered by Bohm \cite{Bohm52}{,} who developed it more systematically (see \cite{DT2009}). 
 We remark that Nelson's ``current velocity field'' coincides
with the guiding velocity field of the de Broglie--Bohm theory. 

 Our treatment of the interacting electron and photon 
combines the ideas of Einstein and of de Broglie--Bohm{,} and implements them in a Lorentz covariant fashion. 
 This is accomplished by adapting to our interacting two-body model the so-called ``hypersurface {Bohm--Dirac}''-type formulation 
for non-interacting particles \cite{DGMZ1999}. 
 This formulation requires a foliation of spacetime by spacelike hypersurfaces, {which 
can be compatible with relativity if it is determined by the wave function in a Lorentz covariant way; cf. \cite{DGNSZ2014}.}
 In our model these hypersurfaces are normal to a time-like Killing vector field $X$ that is determined 
by the initial data of the wave function.
 The guiding field for the particles is furnished by the conserved current of our quantum-mechanical multi-time wave function. 
 We extend a theorem of Teufel--Tumulka \cite{TT2005}, which implies that unique particle motions typically exist globally in time.

 We have carried out some numerical experiments with our system of equations, which indeed demonstrate the process of Compton 
scattering, but which also have revealed an unexpected novel phenomenon: photon capture and subsequent release by the electron.
 The effects of Compton scattering as well as photon capture and release by the electron are nicely illustrated by the numerically 
computed electron and photon trajectories. 
 For many practical purposes such a scenario may well be indistinguishable from the scenario in which a photon is
annihilated and another one created subsequently.
 In our quantum-mechanical two-body model the photon of course never gets destroyed or created, precisely as envisioned by Darwin. 

 By sampling a large ensemble of random positions we also illustrate that the empirical statistics over the possible actual 
trajectories reproduces Born's rule in our model, a consequence of the equivariance of the evolution of the probability densities; 
see Appendix B.

The rest of this paper is structured as follows: 

Section 2 collects the basic mathematical ingredients needed to formulate our model.

In Section 3 we define our two-time electron-photon wave function and system of equations.

Section 4 is devoted to the discussion of our boundary conditions, based on particle current conservation.
 Two propositions and a theorem are stated and proved in this section.

In Section 5 we state and prove the well-posedness of the initial-boundary-value problem for the multi-time wave function, 
according to which a unique global solution exists for each time-like Killing vector $X$. 
 Some technical material is supplemented in Appendix A.

Section 6 identifies the distinguished $X$ which depends only on the initial data for $\Psi^{{(2)}}(\bx_\ph,\bx_\el)$.

Section 7 establishes the hypersurface Bohm--Dirac-type  motion for electron \&\ photon.

Section 8 furnishes a selection of numerically computed electron \&\ photon trajectories. 

In Section 9 we close with an outlook on open questions left for future work.

\section{Preliminaries}
For any $d \in \mathbb{N}$, we let $\bm{\eta} = (\eta_{\mu\nu})$ denote the Minkowski metric on Minkowski spacetime $\Rset^{1,d}$, i.e. 
$$\bm{\eta} = \bm{\eta}^{-1} = \diag(1,\underbrace{-1,...,-1}_{d}).$$
\subsection{Clifford algebra}
Of central importance to the relativistic formulation of quantum mechanics in $d$ space dimensions is the {\em complexified spacetime algebra}
 $\cA$, defined as the complexification of the real Clifford algebra $\mbox{Cl}_{1,d}(\Rset)$ associated with the Minkowski quadratic form of
 signature $(+,-,\dots,-)$.  For $d=1$ this algebra is easily seen to be isomorphic to the algebra of $2\times 2$ complex matrices:
$\cA := \mbox{Cl}_{1,1}(\Rset)_\Cset \cong M_2(\Cset)$.
The isomorphism can be realized by choosing a basis for the algebra:  Let 
\beq
\mathds{1} := \left(\begin{array}{cc} 1 & 0 \\ 0 & 1\end{array}\right),\qquad
\gamma^0 := \left(\begin{array}{cc} 0 & 1\\ 1 & 0\end{array}\right),\qquad 
\gamma^1 := \left(\begin{array}{cc} 0 & -1\\ 1 & 0\end{array}\right),
\eeq
so that we have the Clifford algebra relations
\beq
\gamma^\mu\ga^\nu + \ga^\nu \ga^\mu = 2\eta^{\mu\nu}\mathds{1}.
\eeq
Then a basis for $\mbox{Cl}_{1,1}(\Rset)$ is 
\beq
\cB := \{ \mathds{1}, \ga^0, \ga^1, \ga^0\ga^1 \},
\eeq
and hence
$$\cA = \Span_\Cset \cB = \left\{ \left( \begin{array}{cc} \al+\bet & \ga-\de\\ \ga+\de & \al-\bet\end{array}\right) \right\}_{\al,\bet,\ga,\de \in \Cset} \cong M_2(\Cset).$$ 

\subsection{Lorentz group $O(1,1)$ and its spinorial representation}
The group of isometries of $\Rset^{1,1}$ is the Lorentz group $O(1,1)$.  Viewed as a matrix group, the {\em proper} Lorentz group is identified with $SO(1,1)$, the subgroup of matrices in $O(1,1)$ with determinant 1.  We have
\beq\label{def:SO(1,1)}
SO(1,1) = \left\{ \left(\begin{array}{cc} \cosh a & \sinh a \\ \sinh a & \cosh a \end{array}\right)\right\}_{a\in\Rset}.
\eeq
The full Lorentz group $O(1,1)$ is generated by elements of $SO(1,1)$ together with the space-reflection $\bP := \left(\begin{array}{cc} 1 & 0 \\ 0 & -1 \end{array}\right)$ and time-reversal $\bT := - \bP = \left(\begin{array}{cc} -1 & 0 \\ 0 & 1 \end{array}\right)$ .  

For $\bx \in \Rset^{1,1}$ let 
$$\ga(\bx) := \ga_\mu x^\mu = \left(\begin{array}{cc} 0 & x^0+x^1 \\ x^0-x^1 & 0\end{array}\right) \in \cA$$
be the image of $\bx$ under the standard embedding of the Minkowski spacetime into its Clifford algebra 
(indices are raised and lowered using the Minkowski metric $\bm{\eta}$, and  we are using Einstein's summation convention).
   Let $\bLa$ be an element in $SO(1,1)$, as in \eqref{def:SO(1,1)}.  We then have
$$
\ga(\bLa \bx) = L_{\bLa} \ga(\bx) L_{\bLa}^{-1},
$$
where 
\beq \label{def:Spingroup}
L_{\bLa} = \left(\begin{array}{cc} e^{a/2} & 0 \\ 0 & e^{-a/2} \end{array}\right).
\eeq
It is also easy to check that 
$$
\ga(\bP \bx) = \ga^0 \ga(\bx) \ga^0 = \ga^0 \ga(\bx) \left(\ga^0\right)^{-1},
$$
so that we can set 
$$L_\bP := \ga^0.$$
 We note that $L_{\bLa}$ and $L_\bP$ are unitary operators. 
 When it comes to choosing the operator $L_\bT$ representing time-reversal, it's more advantageous to choose an {\em anti}-unitary operator. 
 It's easy to see that $$ L_\bT := \ga^0\cC$$ works (cf. \cite{ThallerBOOK}, eqs. (3.158-159)), where $\cC$ is the complex conjugation operator
$$
 \cC \psi = \psi^*,\qquad \forall \psi \in M_2(\Cset).
$$
We thus have the spinorial representation of $O(1,1)$ as the group generated by matrices of the form $L_{\bLa}$ as in \eqref{def:Spingroup} 
together with $L_\bP$ and $L_\bT$ defined above.  

The {\em Dirac operator} $D$ on $\Rset^{1,d}$ is by definition
\beq
 D := \ga(\p) = \ga^\mu \frac{\p}{\p x^\mu}
\eeq
Thus for $d=1$ we have
\beq
D = \left(\begin{array}{cc} 0 & \p_0 - \p_1 \\ \p_0+\p_1 & 0 \end{array}\right).
\eeq
Note that $D^2 = \dal \mathds{1}$ where $\dal = \p_0^2 - \p_1^2$ is the one-dimensional wave operator.
\subsection{Spinor fields of ranks two and one}
 A {\em rank-two spinor field} on $\Rset^{1,1}$ is an $\cA$-valued map $\accentset{(2)}{\psi}$ 
that under Lorentz transformations $\bLa\in O(1,1)$ transforms equivariantly with respect to the spinorial representation,
 meaning $\bx' = \bLa \bx$ implies $\phi'(\bx') = L_{\bLa} \phi(\bx) L_{\bLa}^{-1}$. 
Thus {if we denote} 
the components of $\accentset{(2)}{\psi}(\bx)$ as 
$$
\accentset{(2)}{\psi} = \left(\begin{array}{cc} \phi_+ & \chi_- \\ \chi_+ & \phi_- \end{array}\right)
$$
with $\phi_\pm(\bx),\chi_\pm(\bx) \in \Cset$, then under the action of an element of the proper Lorentz group as in \eqref{def:SO(1,1)} these 
components transform as follows
\beq
\phi_+ \to \phi_+,\qquad \chi_+ \to e^{-a}\chi_+,\qquad \phi_- \to \phi_-,\qquad \chi_- \to e^{a}\chi_-,
\eeq
while under the space-reflection $\bP$ we simply have $\phi_\pm \to \phi_\mp,\  \chi_\pm \to \chi_\mp$; and under the time-reversal $\bT$ we 
have $\phi_\pm \to \phi_\mp^*,\  \chi_\pm \to \chi_\mp^*$.  
\begin{rem}\label{rem:diagzero}
{
Note the above transformation rules in particular imply that the {\em diagonal part} of $\accentset{(2)}{\psi}$ transforms like a 
{\em spin-zero field}, so that to isolate the spin-one sector, one needs to project out the diagonal entries in $\accentset{(2)}{\psi}$.}
\end{rem}
To define rank-one spinors, let $\bz \in \Cset^2$ be an {\em isotropic} vector, namely $\bz \ne 0$ and 
$\bm{\eta}(\bz,\bz) = (z^0)^2 - (z^1)^2 = 0$. 
 Let $\bZ = \ga(\bz)$. 
 Then $\det \bZ = 0$ and thus the nullspace of $\bZ$ is non-trivial,
 and one-dimensional since $\bz \ne 0$.
  An element $\accentset{(1)}{\psi} \in \Pset\Cset^2$ is a {rank-one spinor} if there exists an isotropic  vector $\bz$ such that 
$\bZ\accentset{(1)}{\psi} = 0$. Let
$$
\accentset{(1)}{\psi} = \left(\begin{array}{cc} \psi_- \\ \psi_+ \end{array}\right).
$$
It follows that under the action of a proper Lorentz transform of the form \refeq{def:SO(1,1)} the components of a 
{\em rank-one spinor field} on $\Rset^{1,1}$ transform as 
\beq
\psi_- \to e^{a/2} \psi_-,\qquad \psi_+ \to e^{-a/2}\psi_+,
\eeq
while under the space-reflection $\bP$ we simply have $\psi_\pm \to \psi_\mp$ and under the time-reversal $\bT$ we have $\psi_\pm \to  \psi_\mp^*$.

\subsection{One-body wave functions and equations}

\subsubsection{Photon wave function and equation}
 According to \cite{KTZ2018}, in $d$ space dimensions the wave function $\psiPH$ of a single photon is a rank-two bi-spinor field on
 $\Rset^{1,d}$ which, when viewed as a linear transformation, has trace-free diagonal blocks {(see Remark~\ref{rem:diagzero} for why this is the case)}.

  In the case $d=1$, bi-spinors are the same as spinors as defined in the above, while the trace-free condition implies 
$\phi_+ = \phi_- \equiv 0$.
  Thus in one space dimension the wave function of a single photon has only two non-zero components
\beq
\psiPH = \left(\begin{array}{cc} 0 & \chi_- \\ \chi_+ & 0 \end{array}\right).
\eeq
Moreover, according to \cite{KTZ2018} the photon wave function satisfies a Dirac equation with a projection term:
\beq\label{eq:DirPH}
-i\hbar D \psiPH + \mPH\Pi \psiPH = 0,
\eeq
where $\Pi$ is the projection onto diagonal blocks, and $\mPH>0$ to be determined.
  In the case $d=1$, the projection term drops out, so that $\psiPH$ is simply a solution of the {\em massless} Dirac equation:
\beq\label{eq:masslessDir}
-i\hbar \ga^\mu\frac{\p\psiPH}{\p x^\mu} = 0.
\eeq
In components, this becomes
\beq \label{eq:Dircomp}
\left\{\begin{array}{rcl} 
-i\hbar (\p_0-\p_1)\chi_+ & = & 0\\
-i\hbar(\p_0+\p_1) \chi_- & = & 0,
\end{array}\right.
\eeq
which {can be solved as follows.}
 Let $\accentset{\circ}{\psiPH} = \left(\begin{array}{cc} 0 & \accentset{\circ}{\chi_-}\\
 \accentset{\circ}{\chi_+} & 0\end{array}\right)$ be the initial data supplied on the Cauchy hypersurface $\{x^0 = 0\}$ for \refeq{eq:DirEL}.
Let us define the {\em characteristic} coordinates
\beq u := \half(x^0 - x^1),\qquad v := \half(x^0 + x^1).
\eeq
Then \refeq{eq:Dircomp} implies that $\p_u \chi_+ = \p_v \chi_- = 0$, and therefore the solution is 
\beq
\chi_+ = \accentset{\circ}{\chi_+}(x^1+x^0),\qquad \chi_- = \accentset{\circ}{\chi_-}(x^1-x^0).
\eeq
Hence the component $\chi_+$ is constant along {\em left-moving} null rays and $\chi_-$ is constant along {\em right-moving} null rays.
\begin{rem}
{For $d>1$ equation \refeq{eq:DirPH} has a large group of gauge transformations, see \cite{KTZ2018} for details.
 For $d=1$ there is {\em no} such gauge freedom left in the photon wave function $\psiPH$, except for the action of the usual $U(1)$.  
 This is an artifact of being in one space dimension.}
\end{rem}
\subsubsection{Photon probability current}
The existence of a conserved probability current is of profound importance to the understanding of the dynamics of a quantum particle. 
 According to \cite{KTZ2018}, the photon wave function $\psiPH$ described above has an intrinsically defined conserved probability 
current, one which they construct in two steps: {First they show that solutions of the photon wave equation \refeq{eq:DirPH} 
enjoy the following system of conservation laws
\beq\label{def:RieszCons}
\p_\mu\left[ \tr \left(\overline{\phiPH} \ga^\mu \phiPH \ga^\nu\right)\right] = 0,\qquad \nu = 0,\dots,d;\quad d>1,
\eeq
where $\overline{\psi} := \ga^0 \psi^\dag \ga^0$ is the Dirac adjoint for rank-two bispinors, and $\phiPH := \Pi \psiPH$.  It then follows that,}
given any Killing field $X$ of Minkowski spacetime, the manifestly covariant current
\beq\label{def:PHcurr}
j_X^\mu := \frac{1}{{2}} \tr \left(\overline{\phiPH} \ga^\mu \phiPH \ga(X)\right),
\eeq
is conserved, i.e. 
\beq\label{eq:PHconserv}
\p_\mu j_X^\mu = 0.
\eeq
Here $\ga(X) := \ga_\mu X^\mu$. 

 In \cite{KTZ2018} it is next proved that when $X$ is causal  and future-directed, then so is $j_X$, i.e. 
$\bm{\eta}(j_X,j_X)\geq 0$, and $j_X^0 \geq 0$. 
 The authors then show that there exists a distinguished, constant (and therefore Killing) vectorfield $X$ that is 
completely determined by the wave function $\psiPH$ (in fact, given a Cauchy surface $\Si$ it depends only on the 
initial value of $\psiPH$ on $\Si$.)
  They define $X$ {in the following way: Let $\Si$ be any Cauchy surface in Minkowski spacetime, and let $\bn = (n_\nu)$ denote its unit co-normal.
Define
\beq \label{def:pi}
\pi^\mu := \frac{1}{{2}}\int_{\Si}  \tr \left(\overline{\Pi\psiPH} \ga^\mu \Pi\psiPH \ga^\nu\right) n_\nu d\si
\eeq
 The conservation law \refeq{def:RieszCons} then implies that the quantities $\pi^\mu$ are independent of $\Si$, so that if $\{\Si_t\}_{t\in\Rset}$ is
 a foliation of the spacetime by constant-time slices (with respect to an arbitrary time-function $t$), then the quantities $\pi^\mu$ are {\em constant},
  while their manifestly covariant definition \refeq{def:pi} implies that as a $d+1$-component object, with $d>1$, the vector
$\boldsymbol{\pi} := (\pi^\mu)$ transforms correctly, i.e. like a Lorentz vector.}
 Moreover, $\bpi$ is a future-directed causal vectorfield, and is {\em typically} time-like.
  They finally set $X := \boldsymbol{\pi}/\eta(\bpi,\bpi)$  and define  the photon probability current $j_{\mbox{\rm\tiny ph}}$ to be $j_X$ 
for this particular choice of constant vectorfield $X$. 

 Finally, in \cite{KTZ2018} it is shown that $j_{\rm{ph}}$ satisfies the appropriate {version} 
of the Born rule, i.e., if one defines 
\beq
\rho_{\mbox{\rm\tiny ph}} :=
 j_{\mbox{\rm\tiny ph}}^0,\quad v_{\mbox{\rm\tiny ph}}^k := \frac{j_{\mbox{\rm\tiny ph}}^k}{j_{\mbox{\rm\tiny ph}}^0},
\eeq
then one has the continuity equation
\beq
\p_t \rho_{\mbox{\rm\tiny ph}} + \p_k (\rho_{\mbox{\rm\tiny ph}} v_{\mbox{\rm\tiny ph}}^k) = 0.
\eeq
Moreover, in the Lorentz frame where $X = (X^0,0,\dots,0){\in\Rset^{1,d}};\ d>1$, one has 
\beq\label{eq:BornrulePH}
\rho_{\mbox{\rm\tiny ph}} = \frac{ \tr(\Pi\psiPH^\dag \Pi\psiPH) }{\int_{\Rset^d} \tr(\Pi\psiPH^\dag \Pi\psiPH) dx },
\eeq
and thus, $\rho_{\mbox{\rm\tiny ph}}$ is for all practical purposes a probability density, which (for normalized wave 
functions when the denominator in \refeq{eq:BornrulePH} equals 1) depends quadratically on the wave function. 

The proof of \refeq{eq:PHconserv} relies on the fact that for $\psiPH$ satisfying the photon wave equation \refeq{eq:DirPH}, 
its projection onto the diagonal blocks $\phiPH$ satisfies the {\em massless} Dirac equation. 
 As we have seen, for $d=1$ the projection is not needed, the wave function $\psiPH$ itself satisfies the massless Dirac equation 
\ref{eq:masslessDir}. 
Thus in one space dimension we may define the current using the wave function itself:
\beq
j_X^\mu := \frac{1}{{2}} \tr \left(\overline{\psiPH} \ga^\mu \psiPH \ga(X)\right).
\eeq

%%%%%%%%%%%%%%%%%
\subsubsection{Electron wave function and equation}
%%%%%%%%%%%%%%%%%
According to the standard one-body relativistic quantum mechanics (see e.g. \cite{ThallerBOOK}) the wave function of a
 single electron is a rank-one spinor field on $\Rset^{1,1}$,
$$
\psiEL = \left(\begin{array}{c} \psi_- \\ \psi_+ \end{array}\right),
$$
which satisfies a massive Dirac equation
\beq\label{eq:DirEL}
-i\hbar D \psiEL + \mEL \psiEL = 0,
\eeq
where $\mEL>0$ is the electron rest mass. 
 These equations can be written in components as 
\beq\label{eq:Dircomps}
\left\{\begin{array}{rcl} \p_u \psi_+ & = & \frac{-im_\el}{\hbar}\psi_- \\
\p_v\psi_- & = & \frac{-im_\el}{\hbar} \psi_+ \end{array}\right..
\eeq
 These {can be disentangled as follows}:  Let 
$\accentset{\circ}{\psiEL} = \left(\begin{array}{c}\accentset{\circ}{\psi_-}\\ \accentset{\circ}{\psi_+} \end{array}\right)$
 be the initial data for \refeq{eq:DirEL} given on the Cauchy hypersurface $\{x^0 = 0\}$.
  By substituting one equation into the other one in \refeq{eq:Dircomps} we see that both components satisfy the Klein-Gordon equation:
\beq\label{eq:KGs}
\left\{\begin{array}{rcl} \dal \psi_+ + \frac{m_\el^2}{\hbar^2} \psi_+ & = & 0 \\
\psi_+|_{x^0 = 0} &=& \accentset{\circ}{\psi_+} \\
\p_0\psi_+|_{x^0 = 0} & = & \frac{-im_\el}{\hbar} \accentset{\circ}{\psi_-} + \p_1 \accentset{\circ}{\psi_+}\end{array}\right.
\qquad
\left\{\begin{array}{rcl} \dal \psi_- + \frac{m_\el^2}{\hbar^2} \psi_- & = & 0 \\
\psi_-|_{x^0 = 0} &=& \accentset{\circ}{\psi_-} \\
\p_0\psi_-|_{x^0 = 0} & = & \frac{-im_\el}{\hbar} \accentset{\circ}{\psi_+} - \p_1 \accentset{\circ}{\psi_-}.\end{array}\right.
\eeq
{The above equations are solved in the Appendix.}

\subsubsection{Electron probability current}
As is well-known, the electron wave function $\psiEL$ has a conserved probability current, namely 
\beq\label{def:ELcurr}
j_{\mbox{\rm \tiny el}}^\mu := \overline{\psiEL}\ga^\mu \psiEL,
\eeq
where $\overline{\psi} := \psi^\dag \ga^0$ indicates the Dirac adjoint for rank-one bispinors. 
 The manifestly covariant current $j_{\mbox{\rm \tiny el}}$ satisfies 
\beq 
\p_\mu j_{\mbox{\rm\tiny el}}^\mu = 0,
\eeq
and it defines a conserved probability density 
\beq\label{eq:BornruleEL}
\rho_{\mbox{\rm\tiny el}} :=
\frac{ j_{\mbox{\rm\tiny el}}^0 }{\int_{\Rset^d} j_{\mbox{\rm\tiny el}}^0 dx}= \frac{\psiEL^\dag \psiEL}{\int_{\Rset^d} \psiEL^\dag \psiEL dx};
\eeq
{when the denominator in \refeq{eq:BornruleEL} equals 1, $\rho_{\mbox{\rm\tiny el}}$ depends quadratically on the wave function,
compatible with the Born rule.}

\subsection{Multi-time wave functions for an electron-photon system}

In the non-relativistic Schr\"odinger--Pauli ``wave mechanics,'' a two-particle wave function is a map
\beq
{\Phi}: \Rset \times \Rset^d \times \Rset^d \rightarrow \Cset^k,~~~(t, \mathbf{s}_1,\mathbf{s}_2) \mapsto \Phi(t,\mathbf{s}_1,\mathbf{s}_2),
\label{eq:singletimewf}
\eeq
where $d$ is the number of space dimensions and $k$ determined by the number of spin components 
(e.g., if both particles are spin 1/2 electrons and $d=3$, then $k=4$).
 Evidently, such a single-time wave function is not a Lorentz covariant object. 

 For the relativistic treatment of the electron-photon system we need a covariant notion of a many-particle wave function.
 We shall make use of the concept of a multi-time wave function first suggested by Dirac in 1932 \cite{Dir1932} and recently 
developed considerably so as to yield a relativistic version of the Schr\"odinger--Pauli theory (see \cite{LPT2017} for a review).
 A \textit{multi-time (here: two-time) wave function} is a map {(we now purge the superscript ${}^{(2)}$)}
\beq
	\Psi : \mathscr{S} \subset \Rset^{1,d} \times \Rset^{1,d} \rightarrow \Cset^k,~~~(\bx_1,\bx_2)  \mapsto \Psi(\bx_1,\bx_2),
\label{eq:multitimewf}
\eeq
where $\mathscr{S}$ is the set of \textit{space-like configurations}
\beq
	\mathscr{S} = \{ (\bx_1,\bx_2) \in \Rset^{1,d} \times \Rset^{1,d} : {\boldsymbol{\eta}}(\bx_1-\bx_2,\bx_1-\bx_2) < 0 \}.
	\label{eq:spacelikeconfigs}
\eeq  
 One way to visualize $\mathscr{S}$ and other domains in the two-particle configuration space $\Rset^{1,d}\times \Rset^{1,d}$ is to
 use color to distinguish the two events, while depicting them in the {\em same} copy of Minkowski spacetime $\Rset^{1,d}$.
  For example, in one space dimension, using the color blue for electrons and red for photons, a space-like configuration of one 
photon and one  electron can be depicted as in Figure~\ref{fig:spacelike},
\begin{figure}[ht]
\centering
\input{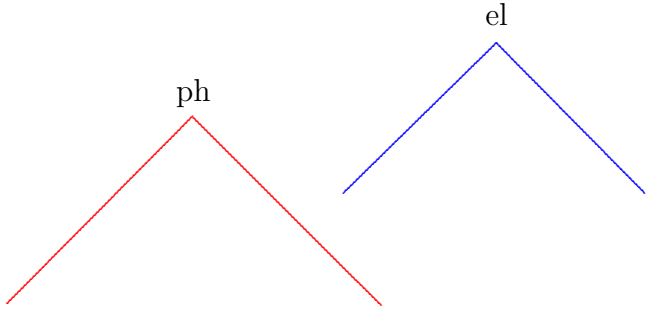}
\caption{\label{fig:spacelike}Example of a space-like configuration}
\end{figure}
while time-like configurations of such a photon-electron pair would be depicted as in Figure~\ref{fig:timelike}.
\begin{figure}[ht]
\centering
\input{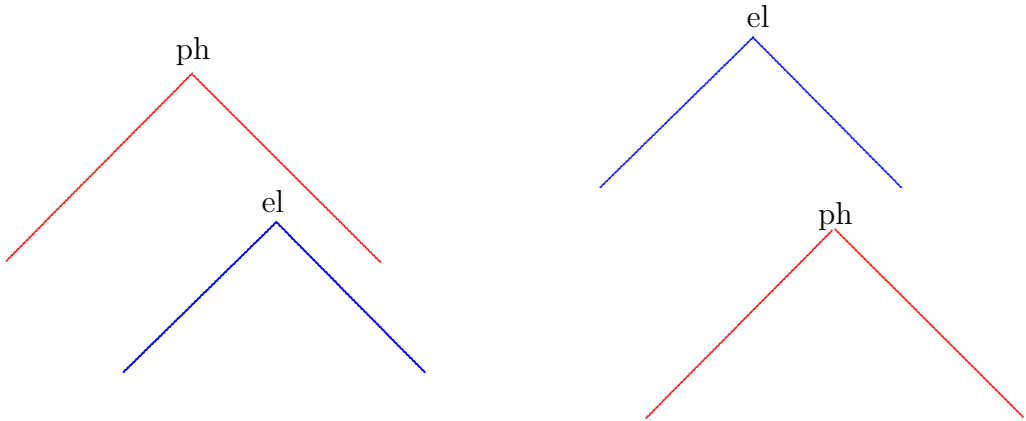}
\caption{\label{fig:timelike}Examples of time-like configurations}
\end{figure}

 Writing $\bx_i = (t_i,\mathbf{s}_i)$ for events in $\Rset^{1,d}$ (we set $c=1${),}
the multi-time wave function $\Psi$ is related straightforwardly to the single-time wave function $\Phi$ by evaluation at equal times 
(relative to a Lorentz frame):
\beq
	{\Phi}(t,\mathbf{s}_1,\mathbf{s}_2) = \Psi(t,\mathbf{s}_1,t,\mathbf{s}_2).
	\label{eq:relationphipsi}
\eeq

As evolution equations, one usually considers a system of two first-order partial differential equations, one for each particle:
\beq
	i \hbar \, \partial_{t_k} \Psi = H_k \Psi,~~k=1,2,
	\label{eq:multitimeeqs}
\eeq
where $H_1,H_2$ are differential operators, called \textit{partial Hamiltonians}. 
 Using \eqref{eq:relationphipsi}, the chain rule implies the usual Schr\"odinger--Pauli equation for 
${\Phi}$, with Hamiltonian $H = H_1+H_2$.

It is assumed that one can rewrite \eqref{eq:multitimeeqs} in a manifestly covariant way, e.g., for Dirac particles:
\beq
	(-i \hbar D_k + m_k) \Psi = 0,~~k=1,2,
	\label{eq:freediracmultitime}
\eeq
where $D_1 = D \otimes \mathds{1}$, $D_2 = \mathds{1} \otimes D$, and $m_1,m_2$ are masses.

 {{We are only interested in multi-time equations \eqref{eq:multitimeeqs} which}} determine $\Psi$ uniquely from initial data on 
$\Sigma^2 := \Sigma \times \Sigma$, where $\Sigma$ is an initial Cauchy surface.
 This is possible if and only if the following \textit{consistency condition} is satisfied \cite{SchweberBOOK,PT2014a}:
\beq
	[i \hbar \, \partial_{t_1} - H_1, i \hbar \, \partial_{t_2} - H_2] = 0.
	\label{eq:cc}
\eeq 
 This condition turns out to be very restrictive. 
 For example, it largely rules out interaction potentials $V_k$ in $H_k = H_k^{\rm free} + V_k(\bx_1,\bx_2)$ \cite{PT2014a,DN2016}.
{On the positive side, \textit{relativistic contact interactions} in 1+1 dimensions 
satisfy the consistency condition \refeq{eq:cc}.
 These can be stated in form of {\emph{intrinsic boundary conditions}} on the two-time wave function
at the coincidence set $\{\bx_1=\bx_2\}$. 
 Interestingly, such boundary conditions can be 
compatible with conservation of particles,
or with particle creation and annihilation \cite{LN2020}; recall our discussion in the introduction. 
 Since we are interested in a fixed number of particles, we choose to work with particle-conserving contact interactions
{which amount to reflective conditions at the intrinsic boundary (illustrated in Fig.~6 below)}.}

 As explained in detail in \cite{Lie2015}, contact interactions can arise naturally from \emph{intrinsic boundary conditions}
of the domain $\mathscr{S}$ for multi-time wave functions.
 Clearly, $\mathscr{S}$ is a domain with boundary, its boundary consisting of {\em lightlike} configurations
\beq\label{def:lightlike}
\mathscr{L} := \{ (\bx_1,\bx_2) \in \Rset^{1,d} \times \Rset^{1,d} : \bx_1 \ne \bx_2,\ \bm{\eta}(\bx_1-\bx_2,\bx_1-\bx_2) = 0 \},
\eeq
together with the set of \textit{coincidence points},
\beq\label{def:coincidences}
	\mathscr{C} = \{ (\bx_1,\bx_2) \in \Rset^{1,d} \times \Rset^{1,d} : \bx_1 = \bx_2 \}.
\eeq
Thus 
$$
\p \mathscr{S} = \mathscr{L} \cup \mathscr{C}.
$$
 As $\mathscr{S}$ is a domain with boundary, the question of boundary conditions arises, for the evolution of the two-time two-body 
wave function in $\mathscr{S}$. 
 We shall see that it is enough to prescribe boundary conditions only on the coincidence set $\mathscr{C}$.  
 Boundary conditions on $\mathscr{C}$ are equivalent to prescribing pair interactions upon contact --- at least
for dynamics in $d=1$ space dimensions.
 At equal times, they then correspond to a $\delta$-potential \cite{LN2015}. 

On the other hand, since the equations satisfied by the wave function are evolution equations,
 it is natural to study the {\em initial value problem} for these equations, which involves 
specifying the values of the unknown on a given initial surface that is space-like with respect 
to both $\bx_1$ and $\bx_2$ variables.
  The introduction of such an initial surface $\Sigma \subset \cM$  divides the domain $\mathscr{S}$ 
of space-like configurations into two subdomains,  $\mathscr{F}_\Sigma$ and $\mathscr{N}_\Sigma$, 
consisting of ``far away" and ``nearby" (with respect to $\Sigma$) configurations, respectively:  
Let $\mathscr{F}_\Sigma$ be the set of configurations $(\bx_1,\bx_2)\in\cM$ such that $\bx_1$ and 
$\bx_2$ are both in the future of $\Sigma$, and the backward light cones emanating from $\bx_1$ and 
$\bx_2$ do {\em not} intersect each other in the future of $\Sigma$.  
Similarly, let $\mathscr{N}_\Sigma$ denote the set of space-like configurations $(\bx_1,\bx_2)$ 
lying in the future of $\Sigma$ with the opposite property, namely their backward light cones do 
intersect in the future of $\Sigma$ (see Figure~\ref{fig:farnear}.)
\begin{figure}[ht]
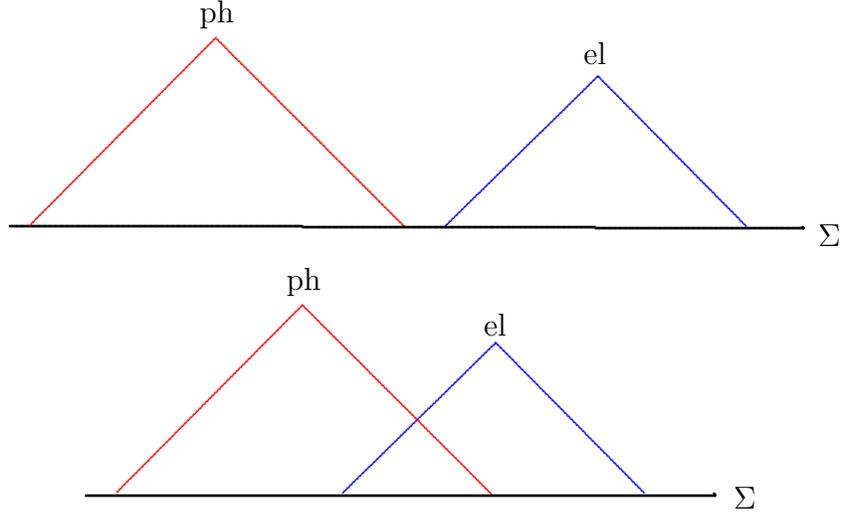

\centering
\input{far_configs.tex}
\input{near_config.tex}
\caption{\label{fig:farnear}Typical configuration in $\mathscr{F}_\Sigma$ (top) and $\mathscr{N}_\Sigma$ (bottom)}
\end{figure}

The significance of this distinction is as follows: since the equations satisfied by the wave function are hyperbolic, they have the 
domain-of-dependence property, namely the value of the wave function at a spacetime point depends only on its initial values on 
$\Sigma$ that lie inside the (closed) backward light cone emanating from that point.
  Thus configurations in $\mathscr{F}_\Sigma$ have the property that the equations in $\bx_1$ and $\bx_2$ can 
be solved independently 
of one another, in other words the two particles corresponding to the $(\bx_1,\bx_2)$ configuration have not 
yet had a chance to interact,
 and are only affected by the initial values.
 By contrast, for the nearby configurations, those in $\mathscr{N}_\Sigma$, the backward light cones of the two particles intersect 
before they reach $\Sigma$ and thus the boundary condition prescribed on the coincidence set $\mathscr{C}$ 
needs to be taken into account.  

The common boundary of the two domains $\mathscr{F}_\Sigma$ and $\mathscr{N}_\Sigma$ is denoted by $\mathscr{B}$.
 We shall see that it plays an important role in solving the initial-boundary value problem for the two-time wave function.

Finally, 
one should note that the multi-time wave function often\footnote{This may require suitable conservation laws, see Sec.~\ref{sec:currents}.}
gives rise to a \textit{probability amplitude for particle detection on Cauchy surfaces} $\Sigma$, given by a generalized version of Born's 
rule \cite{LT2017}.
 This means that a suitable quadratic expression in the multi-time wave function, such as (for two spin-$\half$ Dirac particles)
\beq
	\overline{\Psi}(\bx_1,\bx_2) \,  (\gamma^\mu n_\mu(\bx_1) ) \otimes (\gamma^\nu n_\nu(\bx_2)) \, \Psi(\bx_1,\bx_2),
\label{eq:curvedborn}
\eeq
where $n(\bx)$ is the normal co-vector field at $\Sigma$ and $\overline{\Psi} =\Psi^\dag\ga^0\otimes \ga^0$ is the Dirac adjoint of $\Psi$, 
yields the probability density to detect two particles, at spacetime points $\bx_1,\bx_2 \in \Sigma$.

\section{The  electron-photon two-time wave function and equation}

\subsection{The electron-photon wave function}

In our case, $d=1$ and the configuration spacetime $\Rset^{1,d} \times \Rset^{1,d}$ appearing in \eqref{eq:multitimewf} is given by
$\cM = \Rset^{1,1}_\ph\times \Rset^{1,1}_\el$, where for the rest of this paper, the subscript ${}_\ph$ refers to the photon and ${}_\el$ 
to the electron variables. 
 We denote by $(x_\ph^\mu, x_\el^\nu)$ with $\mu,\nu=0,1$ a global system of rectangular coordinates on $\cM$, and 
refer to space and time coordinates as 
\beq\label{def:coords}
x_\ph^0 = t_\ph,\quad x_\ph^1 = s_\ph,\quad x_\el^0 = t_\el,\quad x_\el^1 = s_\el.
\eeq

The multi-time wave function of the photon-electron system is a section of the vector bundle $\Bset$ over $\cM$ whose fibers are 
isomorphic to the tensor product of the subspace of anti-diagonal matrices in $M_2(\Cset)$ with the vector space $\Cset^2$.  
Thus $\Psi$ has four complex components and can be represented in the form \eqref{eq:multitimewf} with $k=4$. 
A convenient basis for the tensor space can be obtained by taking the tensor product of the bases for each constituent:  
Let 
\beq
\cB_1 := \left\{ E_- = \left(\begin{array}{cc} 0 & 1\\ 0 & 0 \end{array}\right),\quad E_+ = \left(\begin{array}{cc} 0 & 0\\ 1 & 0 \end{array}\right) \right\}
\eeq
be the standard basis for the subspace of anti-diagonal  $2\times 2$ matrices, and let 
\beq
\cB_2 := \left\{ e_- = \left(\begin{array}{c} 1 \\ 0 \end{array}\right),\quad e_+ = \left(\begin{array}{c} 0 \\ 1 \end{array}\right) \right\}
\eeq
be the standard basis for $\Cset^2$.  Then a basis for the tensor product is
\beq
\cB := \left\{ E_-\otimes e_-,E_-\otimes e_+,E_+\otimes e_-,E_+\otimes e_+ \right\},
\eeq
{and t}hus a section $\Psi$ of $\Bset$ will have an expansion
\beq\label{expansion}
\Psi(\bx_\ph,\bx_\el) = 
\sum_{\varsigma_1,\varsigma_2\in\{-,+\}} \psi_{\varsigma_1\varsigma_2}(\bx_\ph,\bx_\el) E_{\varsigma_1}\otimes e_{\varsigma_2},
\eeq
where  {the} four components $\psi_{--},\psi_{-+},\psi_{+-},\psi_{++}$ are complex-valued  functions on $\cM$.  The spinorial character of $\Psi$ reveals itself in the transformation rules of these components: Under a Lorentz transformation\footnote{We are only interested in Lorentz transformations with {\em identical} action on the photonic and electronic variables.} of the form \refeq{def:SO(1,1)} these transform as
\beq
\psi_{--} \to e^{3a/2} \psi_{--},\quad \psi_{-+}\to e^{a/2} \psi_{-+},\quad \psi_{+-} \to e^{-a/2} \psi_{+-},\quad \psi_{++} \to e^{-3a/2}\psi_{++},
\label{eq:tranformationcpts}
\eeq
while under a space-reflection $\bP$ we have 
\beq
\psi_{--} \leftrightarrow \psi_{++},\qquad\psi_{-+}\leftrightarrow\psi_{+-}, 
\eeq
and under a time-reversal $\bT$ we have
\beq
\psi_{--} \to \psi_{++}^*,\quad \psi_{+-} \to  \psi_{-+}^*,\quad \psi_{-+}\to \psi_{+-}^*,\quad \psi_{++} \to \psi_{--}^*,
\eeq
where ${}^*$ denotes complex conjugation extended in the standard way to the tensor space.

A two-body two-time wave function $\Psi$ is called a {\em pure product state} if it is the tensor product of a one-body photon and a 
one-body electron wave function, $\Psi = \psiPH \otimes \psiEL$.
  Equation \refeq{expansion} shows that any two-time two-body wave function can be written as a linear superposition of pure products.
    A two-time two-body wave function that cannot be written as a single pure product state is called {\em entangled}.  

  A two-body quantum system is called {\em interacting} if a pure product initial state can evolve into an entangled state. 

\subsection{Defining equations of the electron-photon wave function dynamics}
We are now ready to write down the defining equations of our photon-electron model.
 Our goal is to set up an interacting dynamics via relativistic contact interactions, as outlined above. Let
\beq
\ga_\ph^\mu := \ga^\mu \otimes \mathds{1},\qquad \ga_\el^\nu := \mathds{1}\otimes \ga^\nu,
\eeq
and
\beq
D_\ph := \ga_\ph^\mu \p_{x_\ph^\mu},\qquad D_\el := \ga_\el^\nu \p_{x_\el^\nu}.
\eeq
Then the multi-time dynamics of the electron-photon wave function is defined by:
\begin{enumerate}
	\item[(a)] The free \textit{multi-time equations} on $\mathscr{S}$:
		\beq\label{eq:joint}
			\left\{ \begin{array}{rcl} -i\hbar D_\ph \Psi & = & 0\\
			-i\hbar D_\el \Psi + m_\el \Psi & = & 0 \end{array}\right.;
		\eeq
	\item[(b)] \textit{Initial data} specified on the surface
\beq\label{def:initialsurf}
\mathcal{I} := \{ (t_\ph,s_\ph,t_\el,s_\el) \in \overline{\mathscr{S}} :  t_\el = t_\ph = 0\},
\eeq
namely
	\beq\label{eq:init}
		\Psi(0,s_\ph,0,s_\el) = \accentset{\circ}{\psi}(s_\ph,s_\el);
	\eeq
	\item[(c)] \textit{Boundary conditions} on $\mathscr{C}$. These implement contact interactions, as described above. 
 Their exact form shall be determined in the next section by considerations about probability conservation. 
 {Here we note} that in 1+1 dimensions, $\mathscr{S}$ decomposes into two disjoint parts, {namely}
\beq\label{def:S12}
\mathscr{S}_1 = 
\{ (t_\ph,s_\ph, t_\el,s_\el) \in \mathscr{S} : s_\ph < s_\el \}, \qquad \mathscr{S}_2 = 
\{ (t_\ph,s_\ph, t_\el,s_\el) \in \mathscr{S} : s_\ph > s_\el \}{,}
\eeq
{and boundary conditions will be linear relations between limits along events in $\mathscr{S}_1$, respectively in
$\mathscr{S}_2$, towards the coincidence set $\mathscr{C}$ (where $s_\ph=s_\el$), of the components of $\Psi$; see 
\refeq{eq:bccomponents}.}
\end{enumerate}

\section{Conserved currents and the boundary condition} \label{sec:currents}

 Given a two-time two-body wave function $\Psi$, we define its {\em Dirac adjoint} $\overline{\Psi}$ as follows: If $\Psi$ is a pure product $\Psi = \psiPH\otimes \psiEL$ then the adjoint is defined by taking the product of the adjoints, i.e.
\beq\label{def:Diradj}
\overline{\Psi} := \overline{\psiPH} \otimes \overline{\psiEL} = \ga^0 \psiPH^\dagger \ga^0 \otimes \psiEL^\dagger \ga^0.
\eeq
For a general two-time two-body wave function, we first write it as a linear superposition of product states and define the 
adjoint by insisting that it be a linear operation. 
\begin{rem}
Despite the appearance of factors $\ga^0$ in \refeq{def:Diradj}, the Dirac adjoint operation is frame-independent.
  On Lorentzian backgrounds, it is the dagger, i.e. conjugate-transpose, operation that is frame-dependent, since 
it requires the choice of a time-like direction to be made (see e.g. \cite{Rie1946}.)
\end{rem}

Given the existence of a conserved probability current for each of the two particles separately 
(i.e. \refeq{def:PHcurr} and \refeq{def:ELcurr},) a good candidate for a joint probability current is the tensor product of the two currents,
 {\em \`a la} \refeq{eq:curvedborn}, for an appropriately chosen {\em constant} vectorfield $X$ that is solely determined by the joint wave function $\Psi$.
   To this end, let $X$ be a fixed {\em time-like}, constant vectorfield on the Minkowski spacetime $\Rset^{1,1}$.
  Later on we will show how to determine $X$ from the wave function. Let
\beq\label{def:jX}
j_X^{\mu\nu}[\Psi] := \frac{1}{{2}} \tr_\ph \{ \overline{\Psi} \ga_\ph^\mu \ga_\el^\nu \Psi \ga_\ph(X) \},
\eeq
where $\tr_\ph = \tr \otimes \mathds{1}$ is the operation of taking the trace of the photonic component (again,
 defined first on pure products and then extended by linearity.
  Recall that the photonic component of the wave function is a rank-two spinor, and a linear transformation, so its trace is well defined.)
We have

\begin{prop}\label{prop:conservedcurrent}
Let $\Psi$ be a {$C^1$-}solution of \refeq{eq:joint} and let $X$ be any  constant vectorfield on $\Rset^{1,1}$.  Then the current $j_X$ is jointly conserved, meaning
\beq \label{eq:continuity}
\left\{\begin{array}{rcl} \p_{x_\ph^\mu} j_X^{\mu\nu}  & =  & 0 \quad\mbox{ for } \nu = 0,1 \\
\p_{x_\el^\nu} j_X^{\mu \nu} & = & 0 \quad\mbox{ for } \mu = 0,1. 
\end{array}\right.
\eeq
\end{prop}
\begin{proof}
The  photon equation and its adjoint, $\partial_{x_\ph^\mu} \overline{\Psi} \gamma^\mu=0$, imply:
\beq
	\partial_{x_\ph^\mu}  j_X^{\mu \nu} = \frac{1}{{2}} \tr_\ph \{ (\underbrace{\partial_{x_\ph^\mu} \overline{\Psi} \ga_\ph^\mu}_{=0}) \ga_\el^\nu \psi \ga_\ph(X) + \overline{\psi}  \ga_\el^\nu (\underbrace{\ga_\ph^\mu\partial_{x_\ph^\mu} \psi}_{=0}) \ga_\ph(X)\}= 0.
\eeq
Moreover, due to the Dirac equation for the electron and its adjoint, we have:
\begin{align}
	\partial_{x_\el^\nu}  j_X^{\mu \nu} &= \frac{1}{{2}} \tr_\ph \{ (\partial_{x_\el^\nu} \overline{\Psi} \ga_\el^\nu) \ga_\ph^\mu \Psi \ga_\ph(X) + \overline{\Psi}  \ga_\ph^\mu (\ga_\el^\nu\partial_{x_\el^\nu} \Psi) \ga_\ph(X)\}\nonumber\\
&= \frac{1}{{2\hbar}} \tr_\ph \{ im_\el \overline{\Psi} \ga_\ph^\mu \Psi \ga_\ph(X) + \overline{\Psi}  \ga_\ph^\mu (-im_\el \Psi) \ga_\ph(X)\}\nonumber\\
&=0.
\end{align}
\end{proof}

We now quote a result from \cite[Thm. 4.4]{Lie2015} which applies to every tensor current $j^{\mu \nu}(\bx_1,\bx_2)$ {that is $C^1$ in $\mathscr{S}$ (but possibly discontinuous across $\mathscr{C}$,) and that satisfies \eqref{eq:continuity} in $\mathscr{S}$.}

\begin{prop}\label{prop:probconserved}
	The {quantity}
	\beq\label{eq:P}
		P := \int_{\Sigma \times \Sigma} d \sigma_\mu(\bx_\ph) \, d\sigma_\nu(\bx_\el)\, j^{\mu \nu}_X(\bx_\ph,\bx_\el)
	\eeq
is conserved ({i.e. it is} independent of the Cauchy surface $\Sigma$) if the following condition of \emph{local {current} 
conservation} is satisfied:
 Let $\epsilon_{\mu \nu}$ stand for the Levi-Civita symbol. Then:
 \beq
	\lim_{\varepsilon \rightarrow 0} \epsilon_{\mu \nu} \left( j_X^{\mu \nu}(t,s-\varepsilon, t, s + \varepsilon) -  
j_X^{\mu \nu}(t,s+\varepsilon, t, s - \varepsilon)\right) = 0~~\forall t,s \in \Rset.
	\label{eq:localprobcons}
 \eeq

{Thus, $\Psi$ can be normalized in such a way that $P\equiv 1$, and then the integrand in \refeq{eq:P} has the interpretation of
 a joint probability density.}
\end{prop}

For simplicity, we specialize to the case
\beq
	\lim_{\varepsilon \rightarrow 0} \epsilon_{\mu \nu}  j_X^{\mu \nu}(t,s-\varepsilon, t, s + \varepsilon) = 0 
= \lim_{\varepsilon \rightarrow 0} \epsilon_{\mu \nu} j_X^{\mu \nu}(t,s+\varepsilon, t, s - \varepsilon).
	\label{eq:localprobcons2}
\eeq

\begin{rem} Physically, this boundary condition corresponds to the electron and the photon not passing through, but completely reflecting 
off of, each other {--- precisely as required for Compton scattering in 1D}.
\end{rem}

\begin{prop}
	Let the vector field $X$ be time-like.
 Then the only translation-invariant class of boundary conditions for the components of $\Psi$ (in any given Lorentz frame) that 
leads to \eqref{eq:localprobcons2} is:
\beq
		\lim_{\varepsilon \rightarrow 0} \psi_{-+}(t,s\mp \varepsilon, t, s \pm \varepsilon) =
 e^{i \theta_\pm} \left( \frac{X^0+X^1}{X^0 - X^1} \right)^{\frac12} \lim_{\varepsilon \rightarrow 0} 
\psi_{+-}(t,s\mp \varepsilon, t, s \pm \varepsilon)~~\forall t,s \in \Rset
	\label{eq:bccomponents}
\eeq
for constant phases $\theta_\pm \in {[0,2\pi)}$.
\end{prop}

\begin{proof}
 	For ease of notation, we omit the limits and arguments of the wave function in \eqref{eq:localprobcons2} and focus on $\mathscr{S}_1$.
 First, we write out $j_X^{\mu \nu}$ in terms of the components of $\psi$. An elementary calculation yields:
\beq
	j_X^{\mu \nu} = 
{\half}\sum_{\rho = 0,1} X_\rho \left[|\psi_{--}|^2+(-1)^\nu |\psi_{-+}|^2+(-1)^{\mu+\rho}|\psi_{+-}|^2+ (-1)^{\mu+\nu + \rho} |\psi_{++}|^2 \right].
\label{eq:jcomponents}
\eeq
Thus:
\begin{align}
\left( j^{01}_X - j^{10}_X \right) ~=~~ &\tfrac{1}{2} \sum_{\rho = 0,1} X_\rho\left[ |\psi_{--}|^2 - |\psi_{-+}|^2 + (-1)^{\rho} |\psi_{+-}|^2   - (-1)^\rho  |\psi_{++}|^2 \right]\nonumber\\
-&\tfrac{1}{2} \sum_{\rho = 0,1} X_\rho \left[ |\psi_{--}|^2 + |\psi_{-+}|^2  - (-1)^\rho |\psi_{+-}|^2  - (-1)^\rho |\psi_{++}|^2 \right]\nonumber\\
=~~&-|\psi_{-+}|^2(X_0 + X_1) + |\psi_{+-}|^2(X_0-X_1),
\end{align}
so that the boundary condition \refeq{eq:localprobcons2} is equivalent to
\beq\label{eq:bndrycnd}
|\psi_{-+}|^2 = \frac{X^0+X^1}{X^0 - X^1}|\psi_{+-}|^2.
\eeq
(Recall that $X_0 = X^0$, $X_1 = -X^1$, and that $X$ is not null.) 
As $X$ is time-like, i.e. $(X^0)^2 - (X^1)^2 > 0$, we have that $\frac{X^0+X^1}{X^0 - X^1}>0$.
 Thus the two sides of \refeq{eq:bndrycnd} can be viewed as the squares of the moduli of two complex numbers, therefore these complex 
numbers can only differ by a phase, which yields \refeq{eq:bccomponents}, except that the phases $\theta_i$ could be functions of $t, s$.
 However, the only translation invariant possibilities clearly are $\theta_\pm = {\rm const}$.
\end{proof}
\begin{thm}
	The only choices for $\theta_\pm$ which make the boundary condition \refeq{eq:bccomponents}  invariant under the full Lorentz group are
\beq
	\theta_-  = -\theta_+ =: \theta ~~~{\rm with}~~~ \theta \in {[0,2\pi)}.
\eeq
\end{thm}

\begin{proof}

We decompose the proof into three steps: 1) invariance under proper Lorentz transformations, 2) reflection invariance, and 3) 
time reversal invariance.

With regard to 1), invariance of \eqref{eq:bccomponents} under $\Lambda \in SO(1,1)$ can be seen using the transformation behavior \eqref{eq:tranformationcpts}:
\begin{align}
	\eqref{eq:bccomponents} ~\stackrel{\Lambda}{\longrightarrow}~~ e^{a/2} \psi_{-+} ~&=~ e^{i \theta_\pm} \left( \frac{ (\cosh a + \sinh a) X^0 + (\cosh a + \sinh a) X^1}{ (\cosh a - \sinh a) X^0 - (\cosh a - \sinh a) X^1} \right)^{\frac{1}{2}} e^{-a/2} \psi_{+-}\nonumber\\
\Leftrightarrow~~~~~ \psi_{-+} &= e^{i \theta_\pm} \left( \frac{e^a( X^0 + X^1)}{ e^{-a} (X^0 - X^1) } \right)^{1/2} e^{-a} \psi_{+-}\nonumber\\
\Leftrightarrow~~~~~ \psi_{-+} &= e^{i \theta_\pm} \left( \frac{X^0 + X^1}{ X^0 - X^1} \right)^{1/2} \psi_{+-},
\end{align}
which gives back \eqref{eq:bccomponents}.

With regard to 2), reflection invariance, we have $\psi_{+-} \leftrightarrow \psi_{-+}$ under a reflection $x \to \bP x$ (as noted before).
 Thus, the boundary condition on $\mathscr{S}_1$ transforms to:
	\begin{align}
		\lim_{\varepsilon \rightarrow 0} \psi_{-+}(t,s-\varepsilon, t, s +\varepsilon) &= e^{i \theta_+} \left( \frac{X^0+X^1}{X^0 - X^1} \right)^{1/2} \lim_{\varepsilon \rightarrow 0} \psi_{+-}(t,s-\varepsilon, t, s + \varepsilon)\nonumber\\
	\stackrel{\bP}{\longrightarrow}~~
	\lim_{\varepsilon \rightarrow 0} \psi_{+-}(t,-s+\varepsilon, t, -s -\varepsilon) &= e^{i \theta_+} \left( \frac{X^0-X^1}{X^0 + X^1} \right)^{1/2} \lim_{\varepsilon \rightarrow 0} \psi_{-+}(t,-s+\varepsilon, t, -s-\varepsilon)\nonumber\\
\Leftrightarrow ~~\lim_{\varepsilon \rightarrow 0} \psi_{-+}(t,-s+\varepsilon, t, -s-\varepsilon) &= e^{-i\theta_+} \left( \frac{X^0+X^1}{X^0 - X^1} \right)^{1/2} \lim_{\varepsilon \rightarrow 0} \psi_{+-}(t,-s+\varepsilon, t, -s -\varepsilon).
	\end{align}
 This is a boundary condition \eqref{eq:bccomponents} on $\mathscr{S}_2$ with
\beq
	\theta_- = -\theta_+.
\eeq
As there is only one phase remaining, we simply call it $\theta = \theta_+ = -\theta_-$.

Now we turn to 3), time reversal invariance.
 Recall that under a time-reversal $\bT$, the relevant components of $\psi$ transform as follows: $\psi_{+-} \to \psi_{-+}^*$ and $\psi_{-+} \to \psi_{+-}^*$.
  Hence
\begin{align}
		\lim_{\varepsilon \rightarrow 0} \psi_{-+}(t,s\mp\varepsilon, t, s \pm\varepsilon) &= e^{\pm i \theta} \left( \frac{X^0+X^1}{X^0 - X^1} \right)^{1/2} \lim_{\varepsilon \rightarrow 0} \psi_{+-}(t,s\mp\varepsilon, t, s \pm \varepsilon)\nonumber\\
	\stackrel{\bT}{\longrightarrow}~~
	\lim_{\varepsilon \rightarrow 0} \psi_{+-}^*(-t,s\mp \varepsilon, -t, s \pm\varepsilon) &=  e^{\pm i \theta} \left( \frac{-X^0+X^1}{-X^0 - X^1} \right)^{1/2} \lim_{\varepsilon \rightarrow 0} \psi_{-+}^*(-t,s\mp\varepsilon, -t,s\pm\varepsilon)\nonumber\\
\Leftrightarrow ~~\lim_{\varepsilon \rightarrow 0} \psi_{-+}(-t,s\mp\varepsilon, -t, s\pm\varepsilon) &= e^{\pm i\theta} \left( \frac{X^0+X^1}{X^0 - X^1} \right)^{1/2} \lim_{\varepsilon \rightarrow 0} \psi_{+-}(-t,s\mp \varepsilon, -t, s \pm\varepsilon)
	\end{align}
which is consistent with the original boundary condition, for all $\theta$.
\end{proof}
The dynamics on the two parts $\mathscr{S}_1, \mathscr{S}_2$ of the domain are now independent, as these are disjoint open sets 
and the boundary conditions for $\Psi$ on one subdomain do not involve the other one. 
This allows us to focus on the dynamics in, say, $\mathscr{S}_1$; the dynamics in $\mathscr{S}_2$ {can then be} treated analogously.

\section{\hspace{-10pt}The initial-boundary-value problem for two-time wave functions}

In this section we consider the initial-boundary value problem (IBVP) for the two-time two-body wave function 
$\Psi = (\psi_{--},\psi_{-+},\psi_{+-},\psi_{++})$,
 consisting of \refeq{eq:joint}, \refeq{eq:init}, and \refeq{eq:bccomponents}, and prove its well-posedness.  

\begin{thm} \label{thm:main}
Let $X$ be any constant time-like and future-directed vectorfield of the Minkowski spacetime $\Rset^{1,1}$, whose components in a given Lorentz frame are denoted $(X^0,X^1)$, and let $\theta \in [0,2\pi)$ be a fixed angle.  Let $\mathscr{S}$ denote the set of space-like configurations in $\cM = \Rset^{1,1}\times \Rset^{1,1}$, as defined in \refeq{eq:spacelikeconfigs}, let $\mathscr{C}\subset \p\mathscr{S}$ denote the set of {coincidence points}, 
 as in \refeq{def:coincidences}, and let $\mathcal{I}$ be the initial surface, as in \refeq{def:initialsurf}.  Let 
$
\accentset{\circ}{\psi}: 
\mathcal{I}\to \Cset^4
$
be {$C^1$} data that is compactly supported in the half-space $\mathcal{I}\cap \overline{\mathscr{S}}_1$ (with $\mathscr{S}_1$ as in \refeq{def:S12}). Assume that the initial data are compatible with the boundary condition, i.e. 
\beq\label{eq:compat}
\opsi_{-+}(s,s) = e^{i \theta} \sqrt{\frac{X^0+X^1}{X^0 - X^1}} \opsi_{+-}(s,s),\qquad \forall s \in \Rset. 
\eeq
 Then the following initial-boundary value problem for the multi-time wave function $\Psi:\cM \to \Cset^4$,
\beq\label{eq:IBVP}
\left\{ 
\begin{array}{rclr} -i\hbar D_\ph \Psi & = & 0 &\\
			-i\hbar D_\el \Psi + m_\el \Psi & = & 0 & \mbox{ in }\mathscr{S}\\
\Psi & = & \accentset{\circ}{\psi} & \mbox{ on }\mathcal{I}\\
 \psi_{-+} & = & e^{i \theta} \sqrt{\frac{X^0+X^1}{X^0 - X^1}} \psi_{+-} & \mbox{ on }\mathscr{C}
\end{array}
\right.
\eeq
(see \refeq{eq:bccomponents} for the precise statement of the boundary condition) has a unique global-in-time solution that is
 supported in $\overline{\mathscr{S}}_1$, depends continuously on the initial data $\opsi$, 
is $C^1$ in $\mathscr{S}_1 \backslash \mathscr{B}$ with 
\beq
 \mathscr{B} := \{ (t_\ph,s_\ph,t_\el,s_\el) \in \mathcal{M} : s_\ph + t_\ph = s_\el - t_\el \},
\eeq
is continuous across $\mathscr{B}$, and its first partial derivatives are bounded in a neighborhood of $\mathscr{B}$.
Furthermore, for every Cauchy surface $\Sigma \subset \mathbb{R}^{1,1}$, the solution is compactly supported in
 $(\Sigma \times \Sigma) \cap \overline{\mathscr{S}}_1$.

\end{thm}
\begin{proof}
We begin by writing the equations in components.  With coordinates on $\cM$ as in \refeq{def:coords}, the massless equations are
\beq\label{eq:massless}
\p_{u_\ph} \psi_{++} = 0,\quad \p_{u_\ph}\psi_{+-} = 0,\quad \p_{v_\ph}\psi_{-+} = 0,\quad \p_{v_\ph}\psi_{--} = 0,
\eeq
while the massive equations read
\bna\label{eq:massive}
\p_{u_\el}\psi_{++} + i \om \psi_{+-} & = & 0 \label{eq:++}\\
\p_{v_\el}\psi_{+-} + i \om \psi_{++} & = & 0 \label{eq:+-}\\
\p_{u_\el}\psi_{-+} + i \om \psi_{--} & = & 0 \label{eq:-+}\\
\p_{v_\el} \psi_{--} + i \om \psi_{-+} & = & 0,\label{eq:--}
\ena
where 
\beq
\om := \frac{m_\el}{\hbar},
\eeq
and we have introduced the null coordinates
\beq
u_\ph = (t_\ph - s_\ph)/2,\quad v_\ph= (t_\ph+s_\ph)/2,\quad u_\el = (t_\el-s_\el)/2,\quad v_\el = (t_\el+s_\el)/2.
\eeq
Our strategy for solving these equations is as follows.  As we have seen before, the massless equations \refeq{eq:massless} imply that $\psi_{++}$ and $\psi_{+-}$ depend only on $(v_\ph; t_\el, s_\el)$ while $\psi_{-+}$ and $\psi_{--}$ depend only on $(u_\ph; t_\el, s_\el)$.   

We have also seen that a pair of the massive equations can be combined to give a Klein-Gordon equation for both components; e.g. \refeq{eq:++} and \refeq{eq:+-} combine to yield
\beq \label{eq:kg+-}
\left\{\begin{array}{rcl}
 (\dal_\el + \om^2)\psi_{+-} &=&0 \\
 \psi_{+-}{\big|_{t_\el = 0}} &=& \opsi_{+-}(s_\ph+t_\ph,s_\el) \\
\p_{t_\el}\psi_{+-}{\big|_{t_\el=0}}& = &-\p_{s_\el} \opsi_{+-}(s_\ph + t_\ph, s_\el) - i \om \opsi_{++}(s_\ph+t_\ph,s_\el).\\
\end{array}\right.
\eeq
and
\beq \label{eq:kg++}
\left\{\begin{array}{rcl}
 (\dal_\el + \om^2)\psi_{++} &=&0 \\
 \psi_{++}{\big|_{t_\el = 0}} &=& \opsi_{++}(s_\ph+t_\ph,s_\el) \\
\p_{t_\el}\psi_{++}{\big|_{t_\el=0}}& = &\p_{s_\el} \opsi_{++}(s_\ph + t_\ph, s_\el) - i \om \opsi_{+-}(s_\ph+t_\ph,s_\el),
\end{array}\right.
\eeq
(where $\dal_\el := \p_{t_\el}^2 - \p_{s_\el}^2$), while \refeq{eq:--} and \refeq{eq:-+} combine to yield
\beq \label{eq:kg--}
\left\{\begin{array}{rcl}
 (\dal_\el + \om^2)\psi_{--} & = & 0 \\
\psi_{--}{\big|_{t_\el = 0}} & = & \opsi_{--}(s_\ph-t_\ph,s_\el) \\
\p_{t_\el}\psi_{--}{\big|_{t_\el=0}} & = & -\p_{s_\el} \opsi_{--}(s_\ph - t_\ph, s_\el) - i \om \opsi_{-+}(s_\ph-t_\ph,s_\el).
\end{array}\right.
\eeq
and
\beq \label{eq:kg-+}
\left\{\begin{array}{rcl}
 (\dal_\el + \om^2)\psi_{-+} & = & 0 \\
\psi_{-+}{\big|_{t_\el = 0}} & = & \opsi_{-+}(s_\ph-t_\ph,s_\el) \\
\p_{t_\el}\psi_{-+}{\big|_{t_\el=0}} & = & \p_{s_\el} \opsi_{-+}(s_\ph - t_\ph, s_\el) - i \om \opsi_{--}(s_\ph-t_\ph,s_\el).
\end{array}\right.
\eeq

In the above Cauchy problems, one may  think of $t_\ph,s_\ph$ as fixed parameters, while the evolution takes place in the $(t_\el,s_\el)$ variables.  This observation allows us to view the above initial value problems as being posed in the {\em projected configuration spacetime}
\beq
\cM_\el(\bx_\ph) := \{ \bx_\el\in\Rset^{1,1}\ |\  (\bx_\ph,\bx_\el)\in \cM\}.
\eeq
Fig.~\ref{fig:regions} show various relevant regions in the projected configuration spacetime $\cM_\el(\bx_\ph)$.
\begin{figure}[ht]
\centering
\input{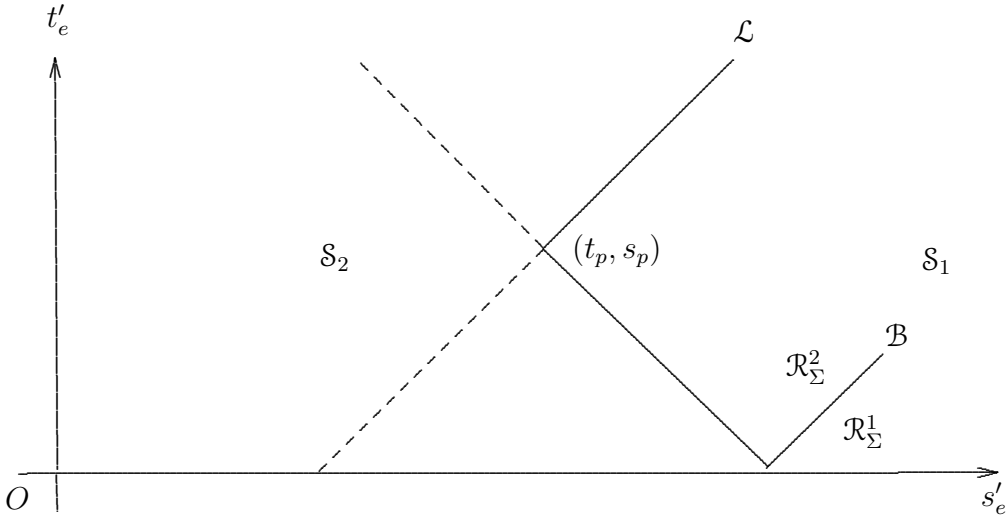}
\caption{\label{fig:regions} Regions in the projected configuration space $\mathcal{M}_\el(\bx_\ph)$}
\end{figure}

Now, these Cauchy problems are easily solvable using \refeq{eq:KGsol}, and as this formula suggests,
 the solution consists of integrating the initial data, against some kernel, on a line segment lying in the initial surface $\cI$. 
Let $\mathbf{q} := (t_\ph,s_\ph,t_\el,s_\el)$ be an arbitrary point in $\mathscr{S}_1$.
  Thus $s_\ph < s_\el$.
 Our objective is to solve for $\psi(\bq)$.
  In order to simplify the presentation, we will henceforth assume $t_\ph,t_\el>0$.
 For \refeq{eq:kg+-} and \refeq{eq:kg++} the line of integration is
\beq 
L_{+} = \{(0,s'_\ph,0,s'_\el)\in \cI\ |\ s'_\ph = s_\ph+t_\ph,    s_\el - t_\el \leq s'_\el \leq s_\el +t_\el \},
\eeq
while for \refeq{eq:kg--} and \refeq{eq:kg-+} that line is
\beq
L_{-} = \{ (0,s'_\ph,0,s'_\el)\in \cI\ |\ s'_\ph = s_\ph - t_\ph,  s_\el - t_\el \leq s'_\el \leq s_\el + t_\el \}.
\eeq
Clearly $L_{-}$ is wholly contained in the region $\mathscr{S}_1$, while $L_{+}$ in general is not.  More specifically, $L_{+}\subset \mathscr{S}_1$ if and only if $s_\ph+t_\ph<s_\el - t_\el$. Let
\beq
\mathscr{R}^1_\Sigma := \mathscr{F}_\Sigma \cap \mathscr{S}_1 =  \{ \bq \in \mathscr{S}_1\ |\ s_\ph + t_\ph < s_\el - t_\el\}
\eeq
denote the configurations corresponding to the particles being ``far" from each other, with respect to $\Sigma$. 
The above analysis means that while \refeq{eq:kg--} and \refeq{eq:kg-+} can be solved to find $\psi_{--}(\bq)$ and $\psi_{-+}(\bq)$ for any $\bq\in \mathscr{S}_1$, this is not the case for \refeq{eq:kg+-} and \refeq{eq:kg++}, which can only be solved using \refeq{eq:KGsol} if the backward characteristics emanating from $\bq$ do not hit the boundary $\mathscr{C}$.  Thus, we can only find $\psi_{+-}(\bq)$ and $\psi_{++}(\bq)$ using this method if $\bq\in \mathscr{R}^1_\Sigma$.  For  configurations that are ``nearby" with respect to $\Sigma$, i.e. those that belong to 
\beq
\mathscr{R}^2_\Sigma := \mathscr{N}_\Sigma \cap \mathscr{S}_1 = \mathscr{S}_1 \setminus \overline{\mathscr{R}^1_\Sigma} = \{ \bq \in \mathscr{S}_1\ |\ s_\ph + t_\ph > s_\el - t_\el\},
\eeq
 we need a different approach.  We will in fact set up and solve a {\em Goursat} problem in $\mathscr{R}^2_\Sigma$ for $\psi_{+-}$. 

The full algorithm for finding the solution $\psi$ is as follows:
Let $\bq = (s_\ph,t_\ph,s_\el,t_\el) \in \mathscr{S}_1$ be fixed.  

{\bf Step 1.} Use \refeq{eq:KGsol} to solve the Cauchy problems \refeq{eq:--} and\refeq{eq:-+} and find $\psi_{--}(\bq)$ and $\psi_{-+}(\bq)$ everywhere in $\mathscr{S}_1$.  The solutions will be as regular as the data, which is assumed to be smooth.

{\bf Step 2.} If $\bq\in \mathscr{R}^1_\Sigma$, use \refeq{eq:KGsol} to solve the Cauchy problems \refeq{eq:kg+-} and \refeq{eq:kg++}, and find $\psi_{+-}(\bq)$ and $\psi_{++}(\bq)$ everywhere in $\mathscr{R}^1_\Sigma$.
Extend the solutions by continuity to the closure of this region, so that $\psi_{+\pm}(s'_\ph,t'_\ph,s'_\el,t'_\el)$ are also known on its boundary $
\mathscr{B}$, where $ s'_\ph+t'_\ph = s'_\el - t'_\el$.

{\bf Step 3.}  Observe that, since the characteristics $u_\ph = \mbox{const.}$ in $\mathscr{S}_1$ intersect the boundary $\mathscr{C}$, the support of $\psi_{-+}$ as it evolves forward in time, will eventually hit the boundary.  Before the first hit, in (\ref{eq:massless}--\ref{eq:--}) we have a non-interacting system, whose solution trivially exists, is unique, and is as smooth as its data.  Without loss of generality therefore, we can assume that the support of the initial data for $\psi_{-+}$ is contained in a coordinate rectangle in the $s'_\ph<s'_\el$ half-plane in $\cI$, one corner of which lies on the diagonal line $s'_\ph = s'_\el$. 
(This can always be done via a suitable translation of the time coordinates $t'_\ph,t'_\el$.)
  
{\bf Step 4.} Let $\bq\in \mathscr{R}^2_\Sigma$ be fixed. We define the characteristic sets
\beq
C^\pm_\bq := \{ \bq'\in\mathscr{S}_1\ |\ s'_\ph+t'_\ph =s_\ph+t_\ph,\ s'_\el \pm t'_\el = s_\el\pm t_\el\}.
\eeq
Since $s_\ph+t_\ph>s_\el - t_\el$, the characteristic $C_\bq^-$ intersects the boundary $\mathscr{C}$ at an earlier positive time $T<\min\{t_\ph,t_\el\}$ (see Figure~\ref{fig:goursat}.)  In fact, let
\beq
T := \half( s_\ph+t_\ph - s_\el + t_\el) > 0,\qquad S:= \half(s_\ph+t_\ph+s_\el - t_\el).
\eeq
Then clearly the ``collision" point is $\bq_c := (T,S,T,S) \in \mathscr{C}\cap C_\bq^-$.  We translate the spatial coordinates $s'_\ph,s'_\el$ in such a way that in the new coordinates, the point $\bq_c$ also belongs to the characteristic  $C^+_\mathbf{0}$:  Let 
\beq
s''_\ph := s'_\ph - (s_\ph+t_\ph),\qquad s''_\el := s'_\el - (s_\ph+t_\ph),
\eeq
and change coordinates on $\cM$ to $(t'_\ph,s''_\ph,t'_\el,s''_\el)$.  In the new coordinates $\bq_c = (T,-T,T,-T) \in C^+_\mathbf{0}$. 
\begin{figure}
\centering
\input{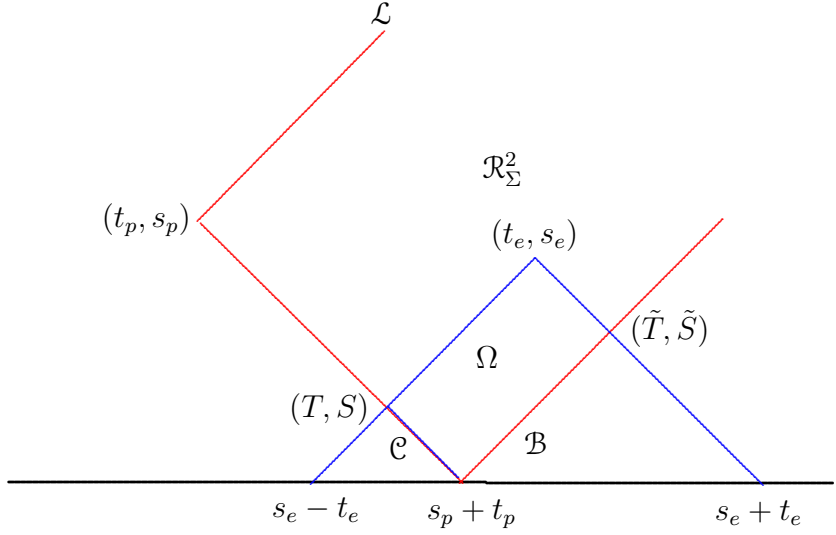}
\caption{\label{fig:goursat}The Goursat problem in region $\mathscr{R}_\Sigma^2$ of the projected configuration space  $\mathcal{M}_\el(\bx_\ph)$}
\end{figure}

 From Step 2 we can compute $\psi_{+-}$ at a point on the boundary $\mathscr{B}$ of the region $\mathscr{R}^1_\Sigma$, where $s''_\ph+t'_\ph = s''_\el - t'_\el = 0$.  Thus we may define a function $\ze:\Rset \to \Cset$ by
\beq \ze(b) := \psi_{+-}(0,0;b,b),\quad \forall b>0.
\eeq

Furthermore, from Step 1, $\psi_{-+}$ is defined everywhere in $\mathscr{S}_1$, in particular on the boundary $\mathscr{C}$.
  Using the boundary condition \refeq{eq:bccomponents} we thus know $\psi_{+-}$ there.
  We may thus define a function $\xi_\theta:\Rset\to \Cset$ by
\beq \xi_\theta(c) := \psi_{+-}(c,-c;c,-c) = 
e^{-i\theta}\sqrt{\frac{X^0+X^1}{X^0-X^1}}\psi_{-+}(c,-c;c,-c).
\eeq
 Therefore, $\psi_{+-}$ is now known on the two characteristic sets $C_\mathbf{0}^\pm$,
allowing us to set up a Goursat problem for $\psi_{+-}$.

{\bf Step 5.}  In the new coordinate system we have chosen, $\bq = (t_\ph,-t_\ph,t_\el,s_\el)$, so that $\bq\in \mathscr{R}^2_\Sigma$ implies that $|s_\el|<t_\el$.  We set up a Goursat problem for a function $U_\theta(t,s)$ defined in the future cone 
\beq
\Om := \{(t,s)\in \Rset^2\ |\ |s| < t\}
\eeq
as follows:
\beq
\left\{ \begin{array}{rclr}
(\p_t^2 - \p_s^2 + \om^2)U_\theta & = & 0&\mbox{in }\Om\\
U_\theta(b,b) & = & \ze(b),& \forall b>0\\
U_\theta(c,-c) & = & \xi_\theta(c),& \forall c>0.
\end{array}\right.
\eeq
Note that, by the compatibility assumption \refeq{eq:compat} on the intial data , we have $\ze(0) = \xi_\theta(0)$.
We solve the above Goursat problem using \refeq{eq:GourKG}, and set
\beq
\psi_{+-}(\bq) = U_\theta(t_\el,s_\el).
\eeq
(Translating back to the old $s'_\ph,s'_\el$ coordinates will reveal the dependence on $t_\ph+s_\ph$.)

{\bf Step 6.} Lastly, now that we have $\psi_{+-}(\bq)$ for all $\bq \in \mathscr{S}_1$, we may recover $\psi_{++}(\bq)$ in the remaining region, i.e. for $\bq\in \mathscr{R}^2_\Sigma$, by integration, via \refeq{eq:++}:
\beq\label{sol:++int}
\psi_{++}(t_\ph,s_\ph,t_\el,s_\el) = \psi_{++}(t_\ph,s_\ph,\tilde{T},\tilde{S}) - i\om\int_{\tilde{T}}^{t_\el} \psi_{+-}(t_\ph,s_\ph,\tau,s_\el+t_\el-\tau)d\tau,
\eeq
where
\beq
\tilde{T} := \half(t_\el +s_\el-t_\ph-s_\ph),\qquad \tilde{S}:= \half(t_\el+s_\el+t_\ph+s_\ph).
\eeq

This completes the algorithm for finding $\psi$ everywhere in $\mathscr{S}_1$.  The resulting formulas, {\em in the original coordinates}, are as follows (recall that these formulas are only valid for $s_\el\geq s_\ph$): 

\bna\label{sol:--}
\psi_{--}(t_\ph,s_\ph,t_\el,s_\el) & = &  \opsi_{--}(s_\ph-t_\ph,s_\el-t_\el)  \\
&&\hspace{-1in}\mbox{} - \frac{\om}{2}\int_{s_\el-t_\el}^{s_\el+t_\el} J_1\left(\om\sqrt{t_\el^2-(s_\el-\si)^2}\right)\sqrt{\frac{t_\el +s_\el - \si}{t_\el-s_\el+\si}}\opsi_{--}(s_\ph - t_\ph, \si) d\si\nonumber\\
&& \hspace{-1in}\mbox{} - \frac{i\om}{2} \int_{s_\el-t_\el}^{s_\el+t_\el} J_0\left(\om\sqrt{t_\el^2-(s_\el-\si)^2}\right) \opsi_{-+}(s_\ph - t_\ph, \si)d\si \nonumber \\
\psi_{-+}(t_\ph, s_\ph, t_\el,s_\el) & = & \opsi_{-+}(s_\ph-t_\ph,s_\el+t_\el) \label{sol:-+} \\
&&\hspace{-1in}\mbox{}-\frac{\om}{2} \int_{s_\el-t_\el}^{s_\el+t_\el}J_1\left(\om\sqrt{t_\el^2-(s_\el-\si)^2}\right) \sqrt{\frac{t_\el - s_\el + \si}{t_\el+s_\el - \si}} \opsi_{-+}(s_\ph - t_\ph,\si) d\si\nonumber\\
&&\hspace{-1in}\mbox{} - \frac{i\om}{2} \int_{s_\el-t_\el}^{s_\el+t_\el} J_0\left(\om\sqrt{t_\el^2-(s_\el-\si)^2}\right) \opsi_{--}(s_\ph - t_\ph, \si) d\si\nonumber
\ena
For $s_\ph+t_\ph \leq s_\el - t_\el$,
\bna\label{sol:+-R}
\psi_{+-}(t_\ph,s_\ph,t_\el,s_\el) & = &  \opsi_{+-}(s_\ph+t_\ph,s_\el-t_\el)  \\
&&\hspace{-1in}\mbox{} - \frac{\om}{2}\int_{s_\el-t_\el}^{s_\el+t_\el} J_1\left(\om\sqrt{t_\el^2-(s_\el-\si)^2}\right)\sqrt{\frac{t_\el+s_\el-\si}{t_\el-s_\el+\si}}\opsi_{+-}(s_\ph+ t_\ph, \si) d\si\nonumber\\
&&\hspace{-1in}\mbox{} -\frac{i\om}{2} \int_{s_\el-t_\el}^{s_\el+t_\el} J_0\left(\om\sqrt{t_\el^2-(s_\el-\si)^2}\right) \opsi_{++}(s_\ph + t_\ph, \si)d\si \nonumber\\
\psi_{++}(t_\ph, s_\ph, t_\el,s_\el) & = & \opsi_{++}(s_\ph+t_\ph,s_\el+t_\el)
 \label{sol:++} \\
&&\hspace{-1in}\mbox{}-\frac{\om}{2} \int_{s_\el-t_\el}^{s_\el+t_\el}J_1\left(\om\sqrt{t_\el^2-(s_\el-\si)^2}\right) \sqrt{\frac{t_\el - s_\el + \si}{t_\el+s_\el - \si}} \opsi_{++}(s_\ph + t_\ph,\si) d\si\nonumber\\
&&\hspace{-1in}\mbox{} - \frac{i\om}{2} \int_{s_\el-t_\el}^{s_\el+t_\el}  J_0\left(\om\sqrt{t_\el^2-(s_\el-\si)^2}\right) \opsi_{+-}(s_\ph + t_\ph, \si) d\si,\nonumber
\ena
while, for $s_\ph+t_\ph > s_\el - t_\el$, setting $\tilde{s}_\el := s_\el - (s_\ph+t_\ph)$ we have
\bna\label{sol:+-QminusR}
\psi_{+-}(t_\ph,s_\ph,t_\el,s_\el) & = & \ze \left(\frac{t_\el+\tilde{s}_\el}{2}\right) + \xi_\theta \left(\frac{t_\el-\tilde{s}_\el}{2} \right) - \half(\ze(0)+\xi_\theta(0))J_0\left(\om\sqrt{t_\el^2 - \tilde{s}_\el^2}\right)\\
&&\mbox{} - \om(t_\el-\tilde{s}_\el) \int_0^{(t_\el+\tilde{s}_\el)/2} \ze(b) \frac{J_1\left(\om\sqrt{(t_\el-\tilde{s}_\el)(t_\el+\tilde{s}_\el-2b)}\right)}{\sqrt{(t_\el-\tilde{s}_\el)(t_\el+\tilde{s}_\el-2b)}} db\nonumber \\
&&\mbox{} - \om(t_\el+\tilde{s}_\el) \int_0^{(t_\el-\tilde{s}_\el)/2} \xi_\theta(c) \frac{J_1\left(\om\sqrt{(t_\el+\tilde{s}_\el)(t_\el-\tilde{s}_\el-2c)}\right)}{\sqrt{(t_\el+\tilde{s}_\el)(t_\el-\tilde{s}_\el-2c)}} dc, \nonumber
\ena
where
\bna
\ze(b) &:=& \psi_{+-}(0,s_\ph+t_\ph,b,s_\ph+t_\ph+b),\label{def:phi}\\
 \xi_\theta(c) & := & e^{-i\theta} \sqrt{\frac{X^0+X^1}{X^0 - X^1}}\psi_{-+}(c,s_\ph+t_\ph-c,c,s_\ph+t_\ph-c),\label{def:xi}
\ena
and the right-hand sides of \refeq{def:phi} and \refeq{def:xi} are to be substituted for from \refeq{sol:+-R} and 
\refeq{sol:-+}, respectively.

Finally, for $s_\ph+t_\ph > s_\el - t_\el$, $\psi_{++}$ can be determined by substituting \refeq{sol:++} and 
\refeq{sol:+-QminusR} into \refeq{sol:++int}.

We turn to the regularity of the solution.  Carefully considering our formulas \eqref{sol:++int}-\eqref{def:xi},
 we note the following points. First of all, the functions $J_0(x)$ and $J_1(x)/x$ which occur in the integration 
kernels are smooth and bounded functions, and the integrals have finite domains.

Now, \eqref{sol:--} shows that $\psi_{--}$ is continuously differentiable in $\mathscr{S}_1$ if
 $\accentset{\circ}{\psi}_{--}$ and $\accentset{\circ}{\psi}_{-+}$ are. Likewise, \eqref{sol:-+} shows that the same condition also 
makes $\psi_{-+}$ continuously differentiable in $\mathscr{S}_1$.

For $s_\ph + t_\ph \leq s_\el - t_\el$, \eqref{sol:+-R} implies that $\psi_{+-}$ is $C^1$ if $\accentset{\circ}{\psi}_{+-}$ and $\accentset{\circ}{\psi}_{++}$ are. Similarly, we can see from \eqref{sol:++} that $\accentset{\circ}{\psi}_{++}$ is $C^1$ in this region under the same condition.\\
For  $s_\ph + t_\ph > s_\el-t_\el$, \eqref{sol:+-QminusR} shows that $\psi_{+-}$ is $C^1$ if $\ze$ 
and $\xi_\theta$ are $C^1$ functions of their argument. By \refeq{def:phi} and \refeq{def:xi} we have just checked that this is the case.
So it is clear that the requirements on the initial data $\accentset{\circ}{\psi}$ ensure that the solution is $C^1$ in $\mathscr{S}_1 \backslash  \mathscr{B}$.

Next, we come to the continuity of $\psi$ across $ \mathscr{B}$. For $\psi_{--}$ and $\psi_{-+}$, there is nothing to show. For $\psi_{+-}$, a short calculation involving the boundary condition for $\accentset{\circ}{\psi}$, \eqref{sol:+-R} yields in the limit $s_\ph+t_\ph \searrow s_\el-t_\el$: $\psi_{+-}(t_\ph,s_\ph,t_\el,s_\el) \rightarrow \psi_{+-}(0,s_\ph+t_\ph,t_\el,s_\el)$ (with $s_\ph+t_\ph=s_\el-t_\el$). As the formula for $\psi(t_\ph,s_\ph,t_\el,s_\el)$ in \eqref{sol:+-R} depends only on $s_\ph + t_\ph$, this agrees with $\psi(t_\ph,s_\ph,t_\el,s_\el)$ in the limit $t_\ph+s_\ph \nearrow s_\ph + t_\ph$.

For $\psi_{++}$, consider \eqref{sol:++int} in the limit $s_\ph + t_\ph \searrow s_\el-t_\el$. Then $\tilde{T} \rightarrow t_\el$ and $\tilde{S} \searrow s_\el$, so \eqref{sol:++int} yields $\lim_{s_\ph + t_\ph \searrow s_\el-t_\el} \psi_{++}(t_\ph,s_\ph,t_\el,s_\el) = \lim_{s_\ph + t_\ph \nearrow s_\el-t_\el} \psi_{++}(t_\ph,s_\ph,t_\el,s_\el)$, i.e. continuity of $\psi_{++}$ across $ \mathscr{B}$.

Finally, the solution formulas also show that the partial derivatives of $\psi$ remain bounded in the limits $s_\ph + t_\ph \nearrow s_\el-t_\el$ and $s_\ph + t_\ph \searrow s_\el-t_\el$ as the partial derivatives $\partial_j \accentset{\circ}{\psi}(s_\ph,s_\el)$ of the initial data remain bounded for $s_\ph \nearrow s_\el$ and $j \in \{ s_\ph, s_\el\}$.

The fact that $\psi$ is compactly supported on any set of the form $(\Sigma \times \Sigma)\cap \overline{\mathscr{S}}_1$ where $\Sigma$ is a Cauchy surface can be read off from the solution formulas as well. In fact, the formulas show that the propagation is not faster than the speed of light in the sense that if $\psi|_{(\Sigma \times \Sigma)\cap \overline{\mathscr{S}}_1}$ is compactly supported in a set $S \subset (\Sigma \times \Sigma)\cap \overline{\mathscr{S}}_1$, then for every other Cauchy surface $\Sigma'$, $\psi|_{(\Sigma' \times \Sigma')\cap \overline{\mathscr{S}}_1}$ is compactly supported in the ``grown set" (compare \cite{LT2017})
\beq
	{\rm Gr}\left(S, (\Sigma' \times \Sigma')\cap \overline{\mathscr{S}}_1 \right) = \left( \bigcup_{(\bx_\ph,\bx_\el) \in S} {\rm LC}(\bx_\ph) \times {\rm LC}(\bx_\el) \right) \cap \left[ (\Sigma' \times \Sigma')\cap \overline{\mathscr{S}}_1\right],
\eeq
where ${\rm LC}(x)$ is the (future and past) solid light cone of $x \in \mathbb{R}^{1,1}$.
\end{proof}

%%%%%%
\section{The probability current}
In the above we have shown the existence of multi-time dynamics for the photon-electron system \refeq{eq:IBVP}, given any constant time-like vector field $X$ on $\Rset^{1,1}$.
Here we propose a particular choice of $X$ that can be constructed entirely from the {\em initial data} $\accentset{\circ}{\psi}$  for the two-body wave function $\psi$. We then show that it can be used to construct a quantum probability current for the system that will be compatible with Born's rule, and will give rise to velocity fields for the two particles in a de Broglie--Bohm setting.

We define, in a given Lorentz frame $\{\be_{(0)},\be_{(1)}\}$ for $\Rset^{1,1}$, and with  $j_X$ as in \refeq{def:jX},
\bna\label{eq:piNULL}
\pi_{(0)} & := & \int_\cI j_{\be_{(0)}}^{00} = \frac{1}{{2}} \int_\cI \tr_\ph(\overline{\opsi}\ga_\ph^0\ga_\el^0 \opsi\ga_\ph^0)  = \frac{1}{{2}}\int_{-\infty}^\infty\int_{-\infty}^\infty \sum_{\varsigma_1,\varsigma_2\in\{+,-\}} |\opsi_{\varsigma_1\varsigma_2}|^2  ds_\ph ds_\el \\
\label{eq:piONE}
\pi_{(1)} & := &  \int_\cI j_{\be_{(1)}}^{00} =\frac{1}{{2}} \int_\cI \tr_\ph(\overline{\opsi} \ga_\ph^0\ga_\el^0\opsi\ga_\ph^1) = \frac{1}{{2}} \int_{-\infty}^\infty\int_{-\infty}^\infty  \sum_{\varsigma_1,\varsigma_2\in\{+,-\}} \varsigma_1|\opsi_{\varsigma_1\varsigma_2}|^2  ds_\ph ds_\el,
\ena
and let 
\beq
\bpi := \sum_{\mu = 0}^1 \pi_{(\mu)} \be_{(\mu)}.
\eeq
It is clear that $\bpi$ is a time-like, future-directed, constant (and therefore Killing) vector field on $\Rset^{1,1}$, 
and that it is constructed entirely from the initial wave function. Let the vector field $X$ be defined by
\beq\label{eq:X}
X := \frac{\bpi}{\eta(\bpi,\bpi)} = \frac{\bpi}{\pi_{(0)}^2 - \pi_{(1)}^2}.
\eeq
$X$ is therefore, like $\bpi$, a time-like, future-directed constant vector field on $\Rset^{1,1}$.

 We now define, following the recipe in \cite{KTZ2018}, the probability current $j$ for the electron-photon system to be 
\beq\label{eq:jmunu}
j^{\mu\nu} := j_X^{\mu\nu} = \frac{1}{2} \tr_\ph (\overline{\psi} \ga_\ph^\mu\ga_\el^\nu \psi \ga_\ph^\al X_\al ).
\eeq
By Prop.~\ref{prop:conservedcurrent}, the current $j$ is conserved under the multi-time flow of \refeq{eq:joint}. 
 Moreover, since the boundary condition \refeq{eq:bccomponents} is satisfied for solutions of \refeq{eq:IBVP}, by 
Prop.~\ref{prop:probconserved}, the corresponding total probability \refeq{eq:P} is conserved on all Cauchy surfaces.
In particular, since
\beq
j^{00} = \frac{1}{{2}\eta(\bpi,\bpi)}\left\{ \pi_{(0)} \sum_{\varsigma_1,\varsigma_2\in\{+,-\}} |\opsi_{\varsigma_1\varsigma_2}|^2 - \pi_{(1)}\sum_{\varsigma_1,\varsigma_2\in\{+,-\}} \varsigma_1|\opsi_{\varsigma_1\varsigma_2}|^2 \right\},
\eeq
we have 
\beq 
\int_{-\infty}^\infty\int_{-\infty}^\infty j^{00} ds_\ph ds_\el = 1
\eeq
at all times, so that we can define
\beq
\rho := j^{00}
\eeq
to be the (joint) probability density of the photon-electron system. More generally, for any smooth Cauchy surface $\Sigma$, 
\beq
\rho_\Sigma(\bx_\ph,\bx_\el) := j^{\mu \nu}_X(\bx_\ph,\bx_\el) \, n_\mu(\bx_\ph) n_\nu(\bx_\el)
\eeq
is the probability density of detecting the photon{-electron pair} on $\Sigma$; 
compare \eqref{eq:curvedborn} (and see \cite{LT2017} for a derivation of this generalized Born rule).

\begin{rem}
 {We note that expressing $X$ in \eqref{eq:jmunu} in terms of $\psi$ with the help of \eqref{eq:X} and 
\eqref{eq:piNULL}, \eqref{eq:piONE}, all factors $1/2$ cancel and $\rho$ takes an analogous form to \eqref{eq:BornrulePH} and
\eqref{eq:BornruleEL}.}
\end{rem}
As an {illustration of the evolution of the joint probability density}, we can solve \refeq{eq:IBVP} with initial data 
corresponding to the wave function being a pure product and its four 
components having the same Gaussian profile, with mean zero for the photon and mean one for the electron. 
 Fig.~\ref{fig:movie} shows plots of the resulting joint probability density at six successive instances of common time $t_\el = t_\ph = t$.
 It is clear that after ``reflection'' from the boundary the wave function is no longer a pure product, so we have a truly interacting 
system.  
\begin{figure}[ht]
\centering
\includegraphics[width=150pt,height=150pt]{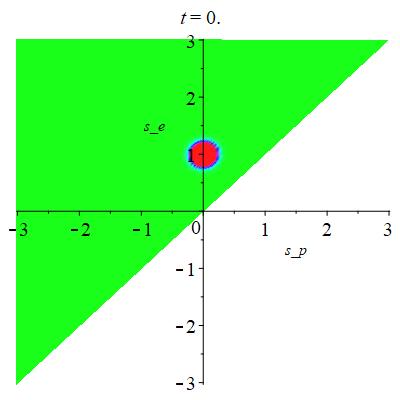}
\includegraphics[width=150pt,height=150pt]{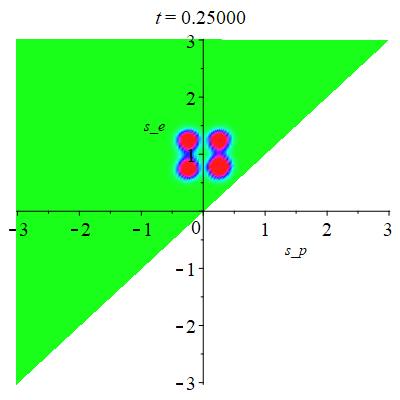}
\includegraphics[width=150pt,height=150pt]{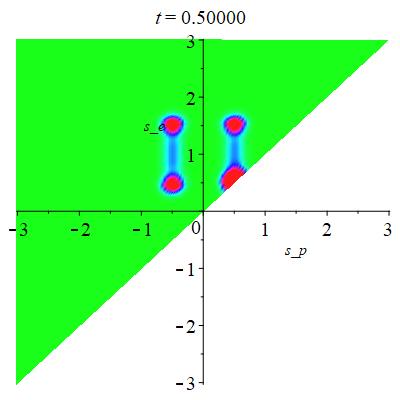}\\
\includegraphics[width=150pt,height=150pt]{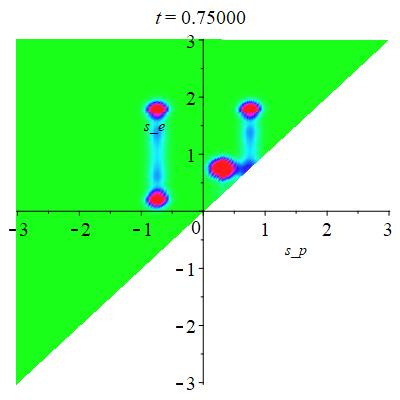}
\includegraphics[width=150pt,height=150pt]{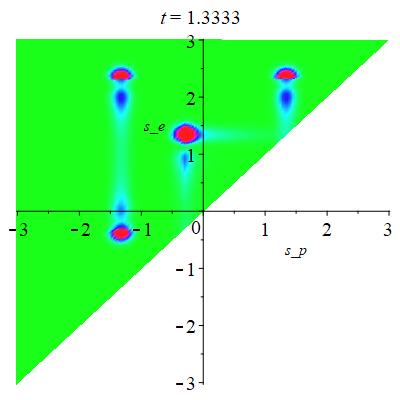}
\includegraphics[width=150pt,height=150pt]{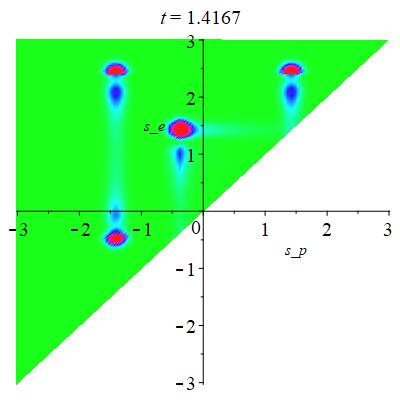}
\caption{Color contour plot of the density $\rho$ in photon-electron configuration space at six consecutive snapshots of common time $t=t_\el=t_\ph$. The photon axis is horizontal, the electron axis vertical. The four ``bumps'' correspond to the four components of the joint wave function.}
\label{fig:movie} 
\end{figure}
\newpage

\section{Electron and photon trajectories} \label{sec:bohmiantrajectories}
 In recent years so-called ``weak measurements'' \cite{AAV} of photon trajectories have become feasible in laboratory experiments 
\cite{KocsisETal}.
 In principle these should be computable from a de Broglie--Bohm type equation of motion for the particle positions.
 The de Broglie--Bohm theory,  a.k.a. Bohmian mechanics (BM) or the pilot wave theory, is a non-relativistic theory which 
stipulates that the building blocks of matter are point particles whose motion is choreographed by Schr\"odinger's, or rather 
Pauli's wave functions in such a way that one recovers the usual statistical quantum predictions from the randomness of the 
initial positions (see \cite{DT2009} for an introduction to BM). 
 A relativistic extension of BM for non-interacting (but possibly entangled) spin-1/2 particles, 
the so-called ``Hypersurface Bohm Dirac (HBD) model," has been proposed in \cite{DGMZ1999}. 
 In the following we show that an HBD-type guiding law can be formulated also for our relativistic 
quantum-mechanical model of an interacting electron-photon system.

 We recall that the law for world-lines 
{of} $N$ electrons (the relativistic analog of particle motion) in the HBD model has two main 
ingredients:
\begin{enumerate}
	\item[(a)] the Dirac tensor current $j^{\mu_1 ... \mu_N}(\bx_1,...,\bx_N) = \overline{\Psi}(\bx_1,...,\bx_N) \gamma^{\mu_1} \otimes \cdots \otimes  \gamma^{\mu_N} \Psi(\bx_1,...,\bx_N)$, and
	\item[(b)] a preferred foliation $\mathscr{F}$ of Minkowski spacetime into space-like surfaces.
\end{enumerate}
 Concerning (a), one can also use a different tensor current $j^{\mu_1...\mu_N}(\bx_1,...,\bx_N)$ to define the world-lines,
 provided its four-divergence with respect to every particle index $\mu_k$ is zero and the contractions of $j$ with $N-1$ 
normal vectors are time-like or light-like and future pointing.
 This may lead to probabilities which are not $|\psi|^2$-distributed, though.
 With regard to (b), a preferred foliation may at first glance seem to contradict the spirit of relativity.
 This point has been critically discussed in \cite{DGNSZ2014}. 
The upshot of the discussion is that the foliation needs to be dynamical, i.e., subject to a law itself, in 
order for the resulting theory to be relativistic in a certain sense.
 Perhaps the most appealing possibility to achieve that is to extract the foliation from the wave function.

 That being said, we readily note that all ingredients for a HBD model are available for our interacting electron-photon model.
 As the tensor current, we use $j^{\mu \nu}_X(\bx_\ph,\bx_\el)$, and the foliation $\mathscr{F} : s \mapsto \Sigma_s$ is given 
by the set of hypersurfaces $\Sigma_s$ orthogonal to $X$.
 The parametrization $s$ of these surfaces is arbitrary.
 Note that $X$ is indeed extracted from the wave function.

Let $Q_\ph(s), Q_\el(s)$ denote the points of intersection of the photon (ph) and electron (el) world-lines with $\Sigma_s$. 
Then the law for their motion is given by:
\begin{align}
	\frac{d Q_\ph^\mu(s)}{ds} ~&=~ j_X^{\mu \nu}(Q_\ph(s), Q_\el(s)) \, X_\nu,\nonumber\\
	\frac{d Q_\el^\nu(s)}{ds} ~&=~ j_X^{\mu \nu}(Q_\ph(s), Q_\el(s)) \, X_\mu.
	\label{hbdlaw}
\end{align}
As the foliation consists of equal-time planes in the frame where $X = (X^0,0)$, it can be simplified as follows in this frame 
(letting $t = Q_\ph^0 = Q_\el^0$):
\beq
	\frac{d Q(t)}{dt} = \frac{J(t,Q(t))}{\rho(t,Q(t))},
	\label{eq:guidance}
\eeq
where $Q = (Q_\ph^1, Q_\el^1)$ and {$J=(J^1,J^2)$, with} $J(t,(s_\ph,s_\el)) = (j^{10}_X(t,s_\ph,t,s_\el), j^{01}_X(t,s_\ph,t,s_\el))$; and where we have set $\rho(t,(s_\ph,s_\el)) = j^{00}(t,s_\ph,t,s_\el)$.

The dynamics thus reduces to an ODE with equal times in the preferred frame given by $X$.
 We shall prove the typical global existence of solutions of \eqref{eq:guidance} on the (timeless) configuration space
\beq
	\mathcal{Q} = \{ (s_\ph, s_\el) \in \mathbb{R}^2 : s_\ph < s_\el\}.
	\label{eq:timelessconfigspace}
\eeq
``Typical" here means for almost all initial configurations $Q_0 \in \mathcal{Q}$ with respect to the measure $\rho \, dq$.
 Such a global existence result was obtained in \cite{BDGPZ1995} for non-relativistic Bohmian mechanics and in \cite{TT2005} 
for general laws of the form \eqref{eq:guidance} where, however, $J$ and $\rho$ were assumed to be $C^1$-functions.
 Thus, we have to make sure that the arguments also work in our case where $J, \rho$ are only Lipschitz continuous.
 The idea of \cite{TT2005} is to devise conditions which exclude all possible ways in which the dynamics may fail to exist, namely:
\begin{enumerate}
	\item The trajectories reach a node of $\Psi$ where $\rho = 0$ and thus \eqref{eq:guidance} may become singular,
	\item The trajectories escape to infinity in finite time,
	\item The trajectories reach a boundary point (where the dynamics is not defined).
\end{enumerate} 
The theorem of \cite{TT2005} is included in Appendix \ref{sec:generalexbm}, see Thm. \ref{thm:generalexbm}.
 We now use it to formulate a more specialized theorem which can be applied to our case.
\newpage

\begin{thm} \label{thm:specializedexbm}
	Let $\mathcal{Q}$ be the configuration space \eqref{eq:timelessconfigspace}.
Moreover, let 
 {$\rho : [0,\infty) \times \overline{\mathcal{Q}} \rightarrow [0,\infty)$ and $J: [0,\infty) \times \overline{\mathcal{Q}} \rightarrow \mathbb{R}^2$} be locally Lipschitz continuous and have the following properties:
	\begin{enumerate}
		\item $\partial_t \rho + {\rm div} \, J = 0$ almost everywhere\footnote{``Almost everywhere'' without any measure specified 
refers to the Lebesgue measure, here and henceforth.},
		\item $\rho > 0$ whenever $\rho \neq 0$,
		\item $\int_{\mathcal{Q}} dq \, \rho(t,q) = 1$ for all $t \geq 0$,
		\item There is a constant $C>0$ such that $|J| \leq C \rho$.
		\item $J$ and the first order derivatives of $\rho$ are bounded on bounded sets (whenever these derivatives exist, i.e., almost everywhere),
		\item $|J^1(t,q) - J^2(t,q)| = 0~ \forall t \geq 0 ~\forall q \in \partial Q$.
	\end{enumerate}
	Let $\mu_t$ be the family of measures on $\mathcal{Q}$ defined by $\mu_t(dq) = \rho(t,\cdot) \, dq$. Then for $\mu_0$-almost all $q \in \mathcal{Q},$ the solution of \eqref{eq:guidance} starting at $Q(0) = q$ exists for all times $t\geq 0$, we have $Q(t) \in \mathcal{Q} \, \forall t \geq 0$ and the family of measures $\mu_t$ is equivariant (see Appendix \ref{sec:generalexbm} for a definition of equivariance).
\end{thm}

\begin{proof}

The idea is to apply Thm. \ref{thm:generalexbm}. We note that the assumption that $\rho, J$ are $C^1$-functions and satisfy 
$\partial_t \rho + {\rm div} \, J = 0$ can, in fact, be relaxed to the requirement that $\rho, J$ are Lipschitz and satisfy 
$\partial_t \rho + {\rm div} \, J = 0$ almost everywhere. (This follows from a careful reading of \cite{TT2005}.)

Now we apply Thm. \ref{thm:generalexbm} to the configuration space $\mathcal{Q}$ from \eqref{eq:timelessconfigspace}. We have: $\mathcal{Q} = \mathbb{R}^2 \backslash S$ where $S = \{ (s_\ph, s_\el) \in \mathbb{R}^2 : s_\ph \geq s_\el \}$. This is an admissible set in the sense of Thm. \ref{thm:generalexbm} with ${\rm dist}((s_\ph,s_\el),S) = |s_\ph-s_\el|$. 
The radial unit vector $\perp\partial S$ pointing in the direction of $S$ is, accordingly, given by $e = \frac{1}{\sqrt{2}} (-1,1)$.

To show the claim we thus need to demonstrate that conditions \eqref{eq:nodecond}, \eqref{eq:escapecond} and \eqref{eq:critsetcond} are satisfied.  Let us first briefly explain the notation used here:  
We define 
\beq
	\mathcal{N}_t := \{ (t,q) \in \mathcal{Q} : \rho(t,q) = 0 \},
\eeq
to be the set of nodes of $\rho$ at time $t$. 
{As $J, \rho$ are assumed to be Lipschitz, we know that away from $\cN_t$, the right-hand-side of \refeq{eq:guidance} is well-defined and locally Lipschitz.}
From standard ODE theory, it follows that for any $q \in \mathcal{Q} \backslash \mathcal{N}_0$, the guiding equation \refeq{eq:guidance} has a unique solution $Q_q(t)$ with $Q_q(0) = q$.  Let  $(\tau_q^-, \tau_q^+)$ be the {\em maximal} time interval of existence for this solution.  
We then extend the configuration space $\cQ$ to $\tilde{\cQ} := \cQ \cup \{\diamondsuit\}$ by adding an extra ``graveyard" configuration $\diamondsuit$, which corresponds to a trajectory failing to exist, namely, for $t\in \Rset$ we set 
\beq
\tilde{Q}_q(t) := \left\{\begin{array}{lr} Q_q(t) & \tau_q^- <t< \tau_q^+\\
\diamondsuit & \mbox{ otherwise}.\end{array}\right.
\eeq
Similarly, if $q\in \cN_0$, we set $\tilde{Q}_q(t):= \diamondsuit$ for all $t \neq 0$, and $\tilde{Q}_\diamondsuit (t) = \diamondsuit$ for all $t\in\Rset$. In this way, we now have a well-defined {\em flow map} $\varphi:\Rset\times\tilde{\cQ} \to \tilde{\cQ}$ and the corresponding one-parameter group of diffeomorphisms $\{\varphi_t\}_{t\in \Rset}$ for \refeq{eq:guidance}, as follows
\beq
\varphi(t,q) := \tilde{Q}_q(t),\qquad \varphi_t(q) := \varphi(t,q).
\eeq
As the sets $\varphi_t(B) \backslash \{\diamondsuit\}$ are not easily accessible, we follow remark 3 below Thm.~1 in \cite{TT2005} and show the conditions {hold} only for all balls around the origin (intersected with $\mathcal{Q}$) instead of all bounded Borel sets of $\mathcal{Q}$. As this replacement enlarges the integrals in \eqref{eq:nodecond}, \eqref{eq:escapecond} and \eqref{eq:critsetcond}, this only yields stronger conditions which imply the ones we need to check.

We first show that the following stronger version of condition \eqref{eq:critsetcond} holds:
	\beq
		\forall r>0: ~~~ \int_0^T dt \int_{\varphi_t(B_r(0)\cap \mathcal{Q}) \backslash \{ \diamondsuit\}} dq ~ \frac{|J(t,q) \cdot e|}{{\rm dist}(q,S)} < \infty.
	\label{eq:critsetcond2}
	\eeq
Here $Q_q(t) := \diamondsuit$ also if $Q_q(t) \in \partial \mathcal{Q}$, i.e. the trajectory reaches a boundary point 
(in the same way as $Q_q(t) := \diamondsuit$ in Thm. \ref{thm:generalexbm} if $Q_q(t)$ reaches a node).
By assumption 4, there is a constant $C>0$ such that $|J| \leq C \rho$. As the velocities in \eqref{eq:guidance} are bounded,
 we have that $\varphi_t(B_r(0)\cap \mathcal{Q}) \backslash \{ \diamondsuit\}$ is contained in $B_{r+Ct}(0)\cap \mathcal{Q}$ for 
every $t \in [0,T]$, in particular in $B_{r+CT}(0) \cap \mathcal{Q}$. We can therefore replace 
$\varphi_t(B_r(0)) \backslash \{ \diamondsuit\}$ with $B_{r+CT}(0) \cap \mathcal{Q}$ in \eqref{eq:critsetcond2}.
 To check that the remaining integral is finite, we use $e = \frac{1}{\sqrt{2}}(-1,1)$, thus $|J(t,q) \cdot e| = \frac{1}{\sqrt{2}} |J^1-J^2|$.
 Moreover, ${\rm dist}((s_\ph,s_\el),S) = |s_\ph-s_\el|$. 
The crucial point now is the {local} Lipschitz property of $J$.
 The function $J^{\rm rel} =(J^1-J^2)/\sqrt{2}$ inherits this property.
 Thus, {for each $R>0$} there is a constant $L\geq 0$ such that 
$|J^{\rm rel}(t,q_1)-J^{\rm rel}(t,q_2)| \leq L \, |q_1-q_2|$ for all {$q_1,q_2 \in B_R(0)\cap \mathcal{Q}$} and all $t\in[0,T]$.
Let $q_1 = (s_\ph, s_\el)$ and choose $q_2 = (s_\ph,s_\ph) \in \partial \mathcal{Q}$.
 Then, as $J^{\rm rel}(q) = 0 \, \forall q \in \partial \mathcal{Q}$ by assumption 6, we find: $ |J(t,(s_\ph,s_\el)) \cdot e| = |J^{\rm rel}(t,(s_\ph,s_\el)| \leq L \, |(s_\ph,s_\el) - (s_\ph, s_\ph)| = L\,|s_\ph-s_\el|$. {We can thus find an $L\geq 0$} such that:
\beq
	\int_0^T dt \int_{\varphi_t(B_r(0)\cap \mathcal{Q}) \backslash \{ \diamondsuit\}} dq ~ \frac{|J(t,q) \cdot e|}{{\rm dist}(q,S)}~ \leq ~T \int_{B_{r+CT}(0)\cap \mathcal{Q}} dq ~ \frac{L \, |s_\ph-s_\el|}{|s_\ph-s_\el|} < \infty.
\eeq
This means that a typical trajectory starting in $\mathcal{Q}$ does not reach the boundary $\partial \mathcal{Q}$.

We now turn to conditions \eqref{eq:nodecond} and \eqref{eq:escapecond}, proceeding similarly as in the proof of corollary 1 in \cite{TT2005}. We shall check the following (stronger) versions of \eqref{eq:nodecond} and \eqref{eq:escapecond}:
\begin{align}
	&\forall r>0: ~~~ \int_0^T dt \int_{\varphi_t(B_r(0)\cap \mathcal{Q}) \backslash \{ \diamondsuit\}} dq \left| \left( \frac{\partial}{\partial t} + \frac{J}{\rho} \cdot \nabla_q \right) \rho(t,q) \right| < \infty,
	\label{eq:nodecond2}\\
 & \forall r>0: ~~~ \int_0^T dt \int_{\varphi_t(B_r(0)\cap \mathcal{Q}) \backslash \{ \diamondsuit\}} dq \left| J(t,q) \cdot \frac{q}{|q|} \right| < \infty.
	\label{eq:escapecond2}
\end{align}
By the bound on the velocities we again have that $\varphi_t(B_r(0)\cap \mathcal{Q}) \backslash \{ \diamondsuit\} \subset B_{r+CT}(0) \cap \mathcal{Q}$. The integrand in \eqref{eq:nodecond2} can be estimated as follows: $\left|  \left( \frac{\partial}{\partial t} + \frac{J}{\rho} \cdot \nabla_q \right) \rho(t,q)\right| \leq \left| \partial_t \rho \right| + C \left| \nabla_q \rho\right|$ almost everywhere. By assumption 5, $J$ and the first order derivatives of $\rho$  (which exist almost everywhere by Rademacher's theorem) are bounded on bounded sets, in particular on the compact set $[0,T] \times \left( \overline{B}_{r+CT}(0) \cap \mathcal{Q}\right)$. Thus, the integral in \eqref{eq:nodecond2} is finite. 
Similarly, the integrand in \eqref{eq:escapecond2} satisfies $\left| J(t,q) \cdot \frac{q}{|q|} \right| \leq |J(t,q)|$ which is bounded on bounded sets. So also the integral in \eqref{eq:escapecond2} is finite.
\end{proof}

\begin{cor}
Let  $\mathcal{Q}$ be given by \eqref{eq:timelessconfigspace}, set {$J(t,(s_\ph,s_\el)) := {(j_X^{10}(t,s_\ph,t,s_\el),j_X^{01}(t,s_\ph,t,s_\el))}$}, and set  $\rho(t,(s_\ph,s_\el)) := j_X^{00}(t,s_\ph,t,s_\el)$ where $j_X^{\mu \nu}$ is given by \eqref{def:jX} and $\Psi$ as in Thm. \ref{thm:main}. The initial data $\opsi$ are assumed to be normalized such that $\int_\mathcal{Q} dq \, \rho(0,q) = 1$.  Let $\mu_t$ be the family of measures on $\mathcal{Q}$ given by $\mu_t(dq) = \rho(t,q) \, dq$. 

Then for $\mu_0$-almost all $q \in \mathcal{Q},$ the solution of \eqref{eq:guidance} starting at $Q(0) = q$ exists for all times $t\geq 0$, we have $Q(t) \in \mathcal{Q} \, \forall t \geq 0$ and the family of measures $\mu_t$ is equivariant.
\end{cor}

\begin{proof}
	We check that the conditions of Thm. \ref{thm:specializedexbm} are met. 
        We begin with the Lipschitz property of $J$ and $\rho$.
	To this end,  we note that $J$ and $\rho$ involve only the components of tensor current $j^{\mu \nu}_X$ which are, according to \eqref{eq:jcomponents}, quadratic expressions in the components of $\Psi$. As a consequence of the regularity properties of $\Psi$ in Thm. \ref{thm:main}, we therefore obtain that also the components of $V$ inherit these properties, i.e. they are
\begin{enumerate}
	\item[(a)] continuously differentiable for $(t_\ph,s_\ph,t_\el,s_\el)  \in \mathscr{S}_1 \backslash  \mathscr{B}$ with $ \mathscr{B} = \{(t_\ph,s_\ph,t_\el,s_\el) \in \mathbb{R}^4: t_\ph +s_\ph = t_\el- s_\el \}$,
	\item[(b)] continuous across $ \mathscr{B}$, and
	\item[(c)] all their partial derivatives are bounded and stay bounded in the limits $t_\ph + s_\ph \nearrow s_\el-t_\el$ and $s_\ph+t_\ph \searrow s_\el-t_\el$.
\end{enumerate}
From these properties, it follows that $J$ and $\rho$ are locally Lipschitz continuous.

Next we check items 1-6 in Thm. \ref{thm:specializedexbm}.
\begin{enumerate}
   \item The continuity equation $\partial_t \rho + {\rm div} \, J = 0$ holds almost everywhere as a consequence of \eqref{eq:continuity}.
 \item {We have $\rho > 0 $ whenever $\rho \neq 0$ as $\rho = \Psi^\dagger \Psi$ (in the frame where $X=(X^0,0)$).}
\item The normalization property $\int_\mathcal{Q} dq \, \rho(t,q) = 1~\forall t \geq 0$ then follows from 
  $\int_\mathcal{Q} dq \, \rho(0,q) = 1$, {the continuity equation on $\mathcal{Q}$, and the boundary condition 
\eqref{eq:localprobcons2} (see also \cite[thm. 4.4]{Lie2015}).}
 \item To show the boundedness of the velocities $J/\rho$, we note that \eqref{eq:jcomponents} yields $j_X^{\mu \nu}/ j^{00}_X \leq 1$.
 This implies: $|J|/\rho = \sqrt{|j^{10}_X|^2 + |j^{01}_X|^2}/ j^{00}_X \leq \sqrt{2} \left(\max \{ |j^{10}_X|, |j^{01}_X|\} \right)/j^{00}_X \leq \sqrt{2}$ (where $j^{\mu \nu}_X$ is to be restricted to equal times).
	\item $J$ and the first order derivatives of $\rho$ are bounded on bounded sets (whenever the first order derivatives exist) as
 follows from the same properties of $\Psi$ (see Thm. \ref{thm:main}).
	\item For all $q = (s,s)\in \partial Q$ and all $t \geq 0$, we obtain using \eqref{eq:localprobcons2}: 
$|J^1(t,q) - J^2(t,q)| = |j_X^{10}(t,s,t,s) - j_X^{01}(t,s,t,s)| = 0$.
\end{enumerate}
Thus, all requirements of Thm. \ref{thm:specializedexbm} are satisfied and the claim follows.
\end{proof}

\section{Numerical Experiments}\vspace{-10pt}
Given an initial wave function $\opsi$, we may solve the initial-boundary-value problem \refeq{eq:IBVP} to find $\Psi$, 
and then compute the corresponding current  {$j^{\mu\nu}_X$} to it. 
  We may then proceed to solve the system of ODEs \refeq{hbdlaw}, with ``typical" (with respect to the initial probability density 
$|\opsi|^2$) data $(Q_\ph(0),Q_\el(0))$ corresponding to the initial actual positions of particles, and plot the resulting trajectories.
  Here we report on the results of our preliminary investigations in this regard.  

We took $\opsi$ to correspond to a pure product of two Gaussian profiles with the same width $\si$ and means that are a distance $d$ apart.
 We randomly chose the amplitudes and phases of the four components, subject to compatibility with the boundary condition \refeq{eq:compat},
 and then normalized $\opsi$ in such a way that the corresponding vectorfield $X$ is equal to $(X^0,0)$.
  Other parameters that need to be chosen are the normalized electron mass $\om$, and the phase angle $\theta$ in the boundary condition.  

As a first test, we computed and compared trajectories in the non-interacting versus the interacting case, 
corresponding to the same initial wave function and the same initial actual configuration; see Fig.~\ref{fig:intvsnonint}.
 The non-interacting case corresponds to ignoring the boundary condition and using the solution formulas 
(\ref{sol:--}--\ref{sol:++}) {\em everywhere} in $\mathscr{S}_1$.  
 For this test we used $\si=0.1$, $d=1$, $\om=2$, $\theta=0$, and took the same initial positions $Q_\el(0)=0.02$, $Q_\ph(0)=0.98$ for 
both the interacting and the non-interacting case.  Note that in the latter case, the particles simply go through each other, while 
in the former, they bounce off of one another, consistent with a Compton scattering scenario.

\begin{figure}[ht]
\centering
\includegraphics[width=170pt,height=170pt]{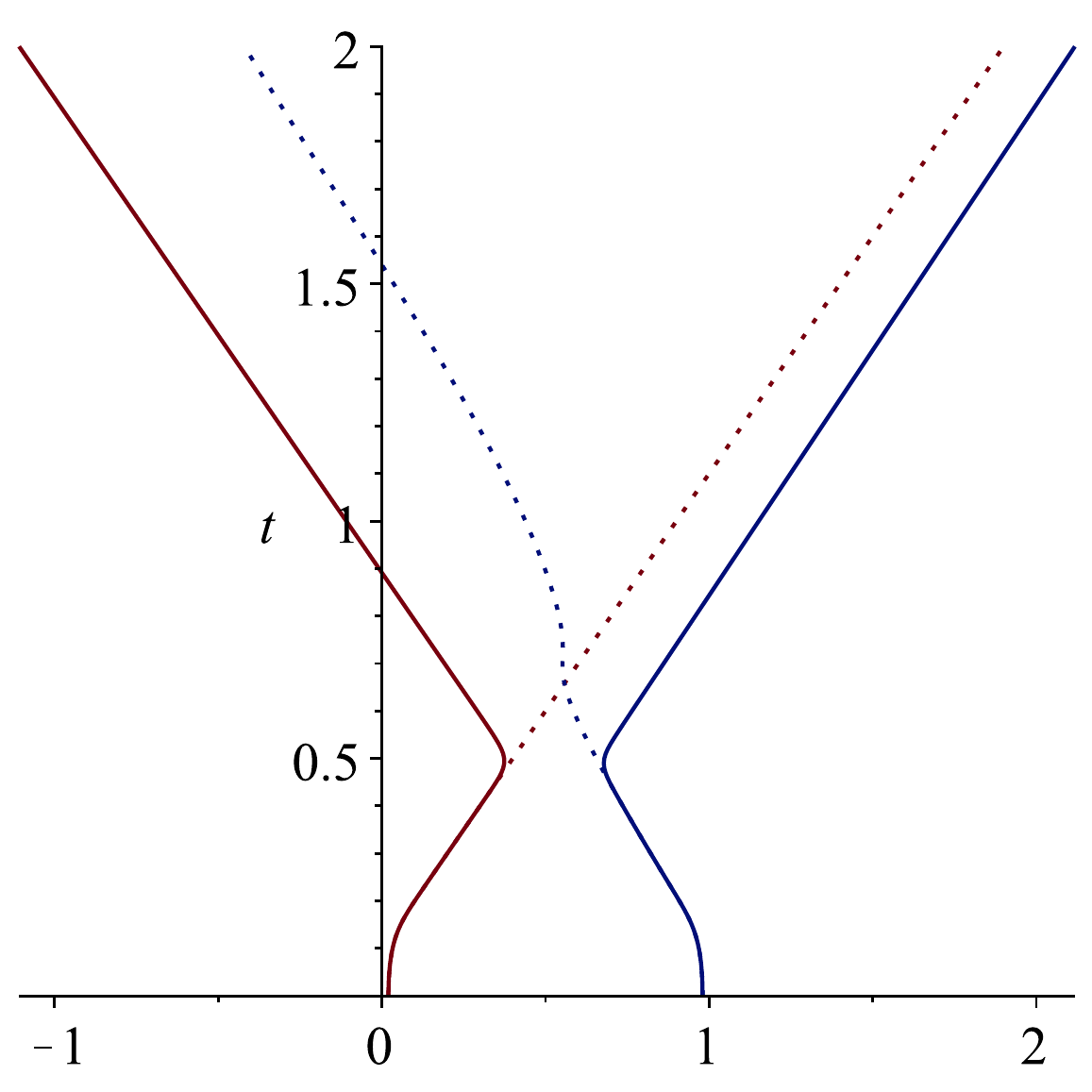}
\caption{\label{fig:intvsnonint} World-lines corresponding to interacting (solid) versus the non-interacting (dotted) electron-photon system,
 with initial photon position at 0.02 and initial electron position at 0.98}
\end{figure}

Next we computed a large number of ``typical" trajectories for the interacting system by randomly choosing the initial actual 
configuration according to the initial density $|\opsi|^2$. 
 We used $\si=0.1$, $d=1$, $\om=2$ and $\theta = 0$;
see Fig.~\ref{fig:bohmtraj}, {where different colors indicate different pairs of trajectories.} 
 Note that consistent with the evolution of the density $\rho$ (cf. Fig.~\ref{fig:movie}), 
at late times the distribution of actual configurations develops four peaks, 
corresponding to the particles: 
(1) always moving away from each other, (2) both moving to the right, (3) both moving to the left, (4) initially approaching 
each other, then bouncing off of one another and moving away (scattering).
\begin{figure}[ht]
\centering
\includegraphics[width=180pt,height=180pt]{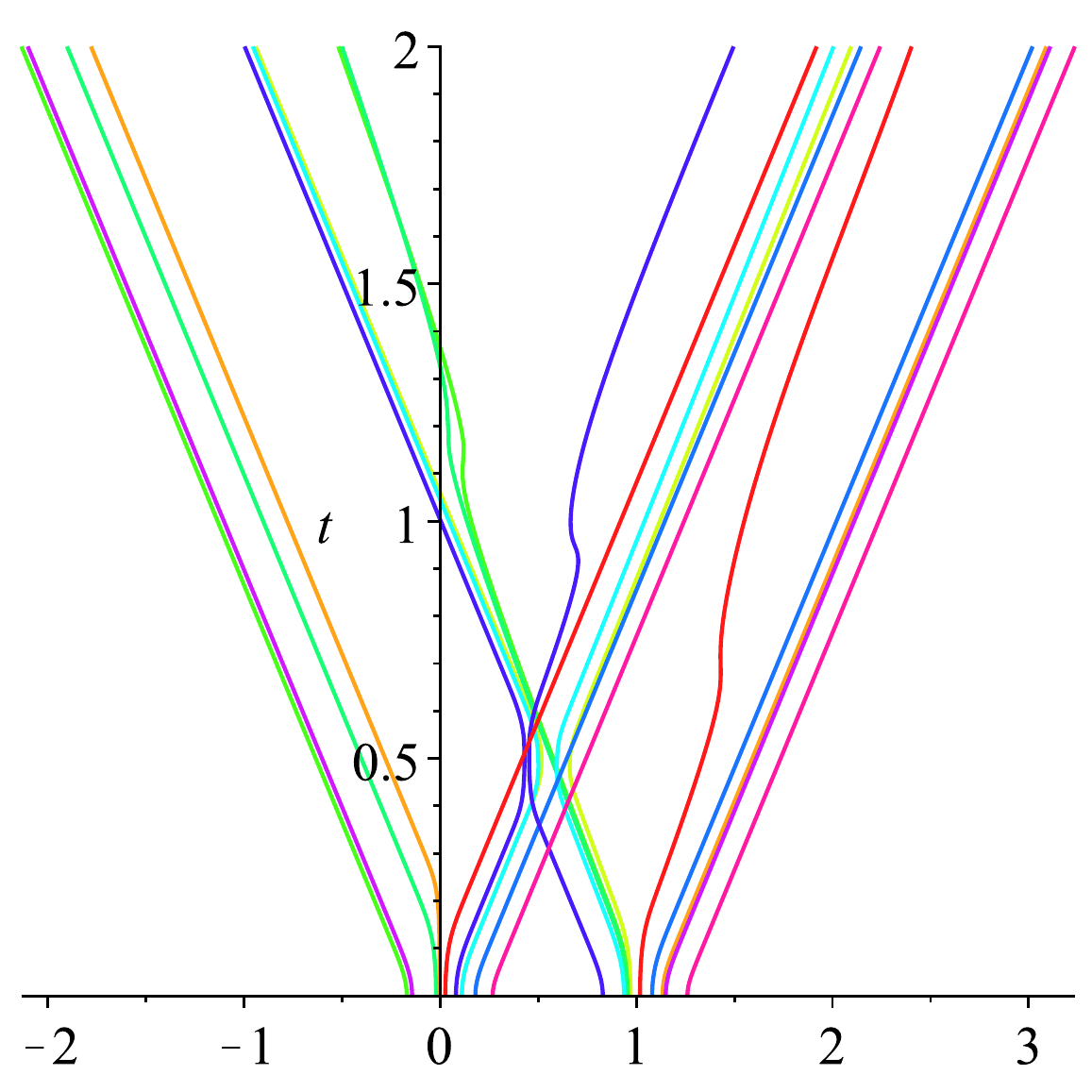}\qquad\qquad
\includegraphics[width=180pt,height=180pt]{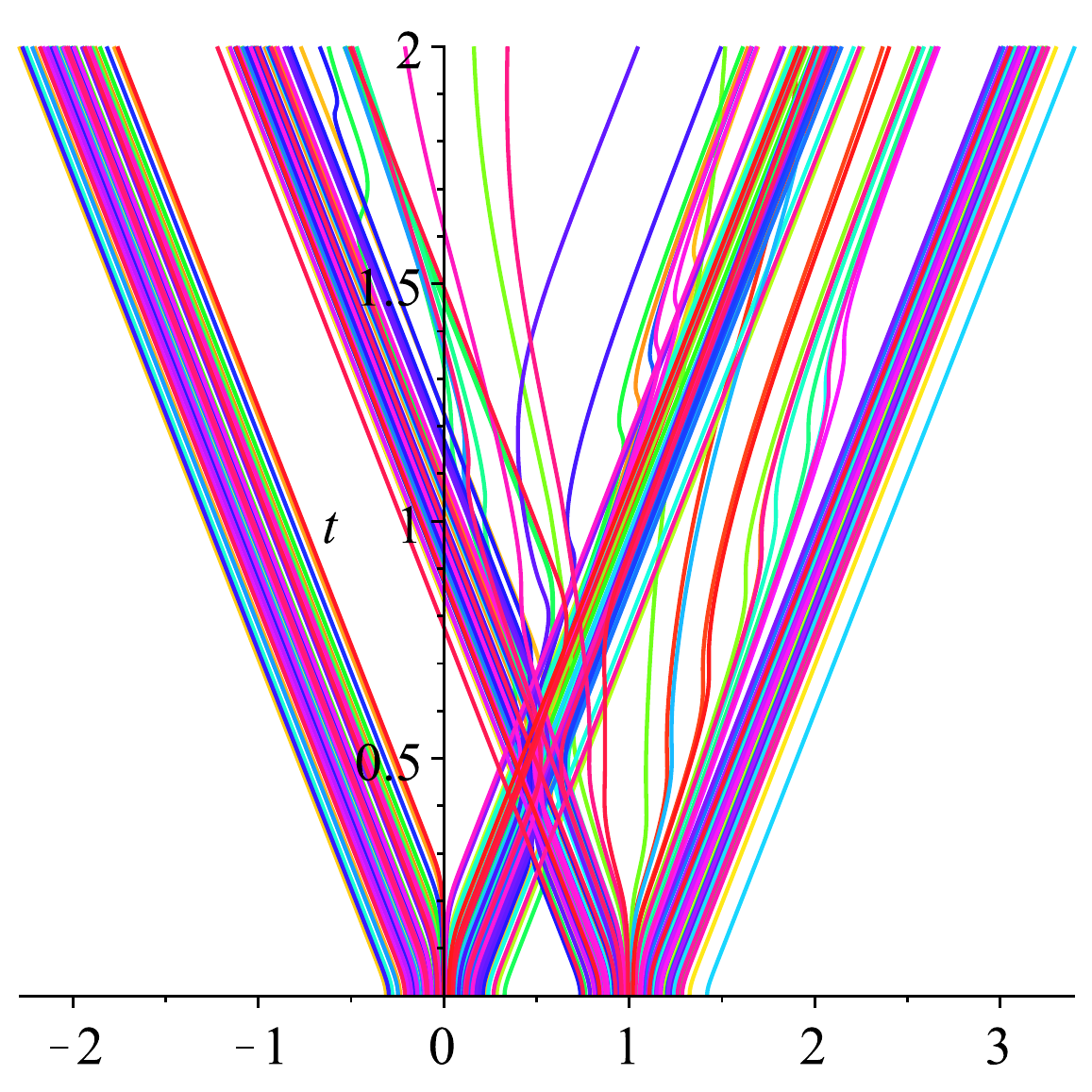}
\caption{\label{fig:bohmtraj}Electron \&\ photon world-lines corresponding to 10 pairs (on the left) and 100 pairs (on the right) of
 typical initial positions. The photon is to the left of the electron.}
\end{figure}

 Also of note is the existence of trajectories where the photon and electron approach each other and then
spend an appreciable amount of time together in very close proximity, almost as in a bound or quasi-bound state.
 Perhaps not too surprisingly, this phenomenon appears to become more typical when the electron is given a larger mass. 
 See Fig.~\ref{fig:romance} where six pairs of trajectories have been shown for $\om=1$.  
\begin{figure}[ht]
\centering
\includegraphics[width=200pt,height=200pt]{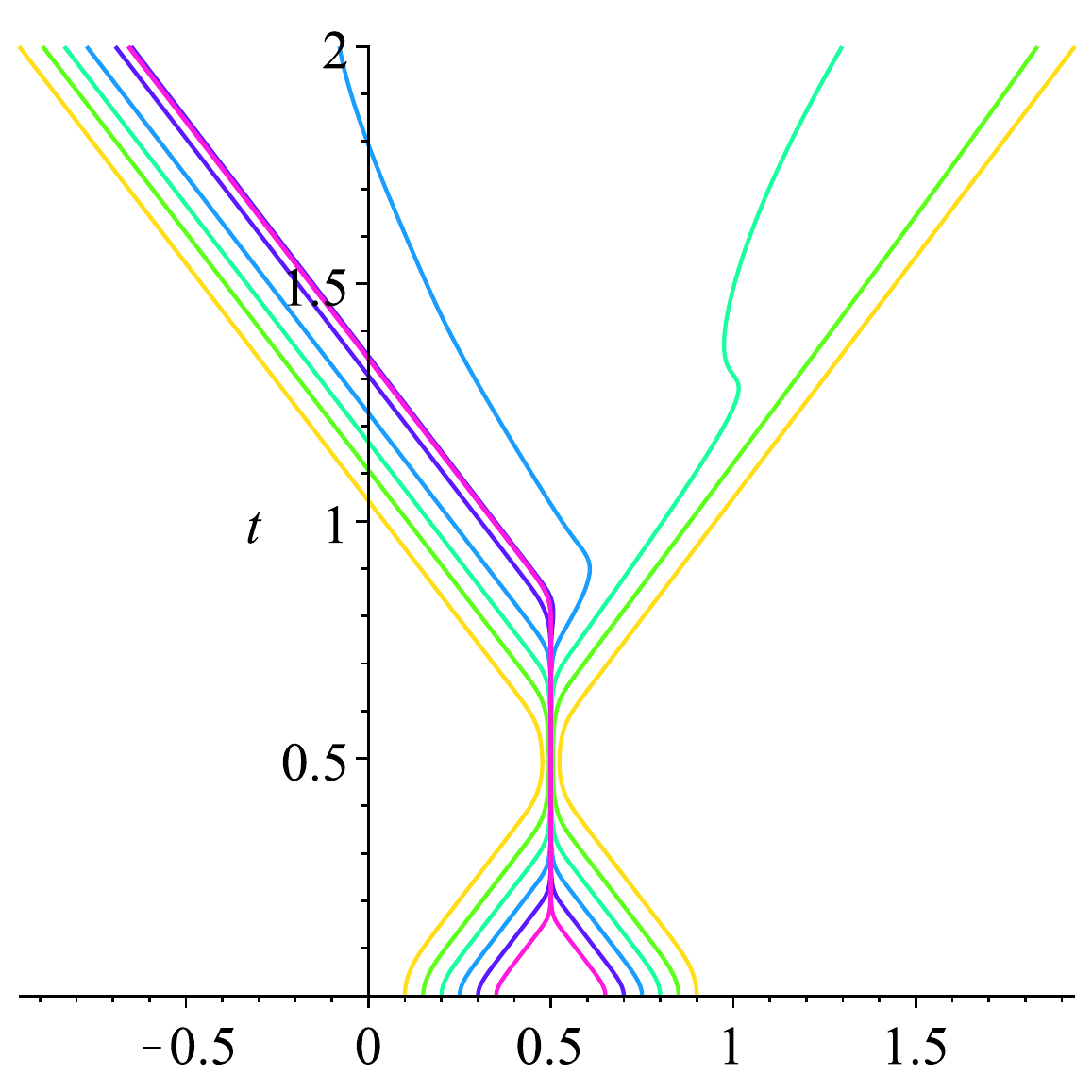}
\caption{\label{fig:romance}World-lines corresponding to possible bound or quasi-bound states}
\end{figure}
 Note in particular the pair of trajectories whose initial positions are closest to each other, which seem to indicate a bound state 
between electron and photon may have formed.  
 This may be hardly distinguishable from a scenario in which the photon gets annihilated, transferring its energy-momentum to the electron. 
 In our model the photon of course does not get annihilated at all. 
 We will speak of ``photon capture'' by the electron. 
 The time-reversed process is possible as well, and we will speak of ``photon release.''

 Finally, we also tested for the dependence of the trajectories on the parameter $\theta$.
  Our numerical tests so far have {\em not} shown any discernible difference between trajectories corresponding to 
different values of $\theta$, once we factor out the dependence of the initial wave function on this parameter as 
required by the compatibility condition \refeq{eq:compat}.
  (Recall that one way to satisfy this compatibility condition is for intial wave function to be zero on the coincidence set, 
so it is possible to choose initial data that are independent of $\theta$.)

\section{Summary and Outlook}

 After having demonstrated in this paper that a relativistic quantum-\emph{mechanical} treatment of the Compton effect
as a two-body problem is feasible in 1+1 spacetime dimensions, the next goal is to find out whether the same 
can be accomplished in 3+1 spacetime dimensions, as envisioned by C.G. Darwin (recall the quotation at the beginning of our paper).
 The larger goal, of course, is to eventually find out whether empirical electromagnetism in 3+1 spacetime dimensions
can be accurately accounted for in terms of a relativistic quantum mechanics with a \emph{fixed} finite number $N$ of 
electrons, photons, and their anti-particles.

 Furthermore, following the lead of \cite{DGMZ1999} we have formulated an interacting generalization of their
``Hypersurface Bohm--Dirac (HBD) model."
 In our theory, both electrons {\em and} photons are treated as point particles whose positions move, guided in a deterministic
manner by the quantum-mechanical multi-time wave function.
 This should help  alleviate fears that the non-relativistic de Broglie--Bohm type theory of interacting particles 
could not be made relativistic.
 Of course, our 1+1 dimensional formulation is only a first step, but
some of us are optimistic that a 3+1-dimensional formulation could be feasible.
\newpage

\subsection*{Acknowledgments}
  We would like to thank Hans Jauslin, Stefan Teufel, and Roderich Tumulka for helpful discussions.  Thanks to Lawrence Frolov and Samuel Leigh for giving our paper a careful reading and for their helpful suggestions.
  Thanks go also to the referees for their comments. 
\\[1mm]
\begin{minipage}{15mm}
 \includegraphics[width=13mm]{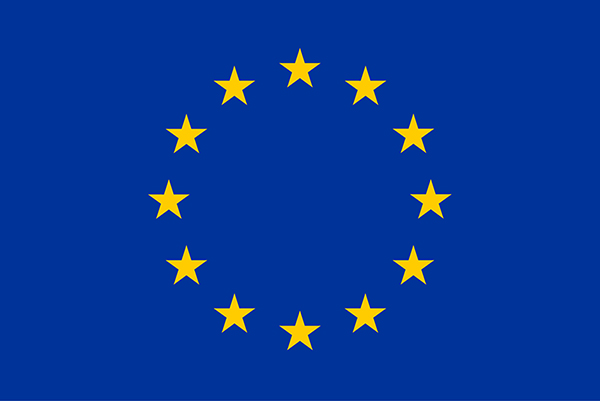}
\end{minipage}
\begin{minipage}{148mm} % This phrase is prescribed by the European Research Council.
This project has received funding from the European Union's Framework for
Research and Innovation Horizon 2020 (2014--2020) under the Marie Sk{\l}odowska-Curie Grant Agreement No.~705295.
\end{minipage}
\bigskip

 On behalf of all authors, the corresponding author states that there is no conflict of interest.

\appendix 
\bigskip
\centerline{\bf \large{APPENDIX}}

\section{Solution formulas for the Klein-Gordon equation}
We recall that the general solution to the Cauchy problem for the Klein-Gordon equation 
\beq
w_{tt} - w_{ss} + w = 0,\qquad w(0,s) = f(s),\qquad w_t(0,s) = g(s),
\eeq
is
\bna
w(t,s) & = & \half\{ f(s-t)+f(s+t) \}- \frac{t}{2}\int_{s-t}^{s+t} \frac{J_1(\sqrt{t^2 - (s-\si)^2})}{\sqrt{t^2 - (s-\si)^2}} f(\si) d\si \nonumber\\
&&\mbox{} + \half\int_{s-t}^{s+t} J_0(\sqrt{t^2 - (s-\si)^2}) g(\si) d\si ,\label{eq:KGsol}
\ena
where $J_\nu$ is the Bessel function of index $\nu$.

We also need the solution formula for the Goursat (i.e. characteristic initial data) problem for the 1-dimensional Klein-Gordon equation:
\beq
U_{tt} - U_{ss} + U = 0\mbox{ for }|s|<t,\qquad U(b,b) = \ze(b),\qquad U(c,-c) = \xi(c),
\eeq
which is (see e.g. \cite{Gar1998}, eq. (4.85)):
\bna
U(t,s) & = & \half (\ze(0)+\xi(0))J_0(\sqrt{t^2-s^2})\nonumber \\
&&\mbox{} +   \int_0^{(t+s)/2} \ze'(b) J_0(\sqrt{(t-s)(t+s-2b)}) db\label{eq:GourKG}\\
&&\mbox{} +
\int_0^{(t-s)/2} \xi'(c) J_0(\sqrt{(t+s)(t-s-2c)}) dc \nonumber
\ena
(Note that $\ze(0)=\xi(0)$ is necessary for continuity.) Using integration by parts, the above can also be written as
\bna
U(t,s) & = &  \ze \left(\frac{t+s}{2}\right) + \xi \left(\frac{t-s}{2}\right) - \half (\ze(0)+\xi(0))J_0(\sqrt{t^2-s^2}) \nonumber\\
&&\mbox{} - (t-s) \int_0^{(t+s)/2} \ze(b) \frac{J_1(\sqrt{(t-s)(t+s-2b)})}{\sqrt{(t-s)(t+s-2b)}} db\nonumber \\
&&\mbox{} - (t+s) \int_0^{(t-s)/2} \xi(c) \frac{J_1(\sqrt{(t+s)(t-s-2c)})}{\sqrt{(t+s)(t-s-2c)}} dc.
\ena

%%%
\newpage

\section{Existence and uniqueness for the electron-photon trajectories} \label{sec:generalexbm}

For completeness, we include the theorem about the global existence and uniqueness of Bomian trajectories by Teufel and Tumulka \cite{TT2005}
 which we heavily use in Sec. \ref{sec:bohmiantrajectories}.

The notation is as follows: 
The configuration space $\mathcal{Q}$ is taken to be either $\mathcal{Q} = \mathbb{R}^d$ or 
$\mathcal{Q} = \mathbb{R}^d \backslash \bigcup_{l=1}^m S_l$ where each $S_l$ is an admissible set. 
A set $S \subset \mathbb{R}^d$ is said to be admissible if there is a $\delta > 0$ such that the distance function 
$q \mapsto \text{dist}(q,S)$ is differentiable on the open set $(S + \delta) \backslash S$ where 
$S + \delta = \{ q \in \mathbb{R}^d : {\rm dist}(q,S) < \delta\}$. Let
$\cN_t$, $Q_q(t)$, and $\varphi_t$ be defined as before. 
 This allows us to introduce the notion of \textit{equivariance} as follows: Let $\rho_t$ be the distribution of $Q_q(t)$ if $q$ has 
distribution $\mu_0$, i.e., $\rho_t = \mu_0 \circ \varphi_t^{-1}$. One then says that \textit{the family of measures $\mu_t$ is 
equivariant on a time interval $I$} if $\rho_t = \mu_t$ for all $t \in I$. (Intuitively, this means that the measures $\mu_t$ 
capture the time-evolved distribution of trajectories correctly.)

\begin{thm}[{\cite[thm. 1]{TT2005}}] \label{thm:generalexbm}
Let $\mathcal{Q} \subset \mathbb{R}^d$ be a configuration space as defined above and let $j=(\rho, J)$ be a current with 
$\rho : [0,\infty) \times \mathcal{Q} \rightarrow [0,\infty)$, $J :  [0,\infty) \times \mathcal{Q} \rightarrow \mathbb{R}^d$ where:
	\begin{enumerate}
		\item[(i)] $\rho$ and $J$ are continuously differentiable,
		\item[(ii)] $\partial_t \rho + {\rm div} \, J = 0$,
		\item[(iii)] $\rho > 0$ whenever $\rho \neq 0$,
		\item[(iv)] $\int_\mathcal{Q} dq \, \rho(t,q) = 1$ for all $t \geq 0$.
	\end{enumerate}
	Let $T>0$ and let $\mathscr{B}(\mathcal{Q})$ denote the set of all bounded Borel sets in $\mathcal{Q}$. Suppose that
	\beq
	\forall B \in \mathscr{B}(\mathcal{Q}): ~~~ \int_0^T dt \int_{\varphi_t(B) \backslash \{ \diamondsuit\}} dq
 \left| \left( \frac{\partial}{\partial t} + \frac{J}{\rho} \cdot \nabla_q \right) \rho(t,q) \right| < \infty,
	\label{eq:nodecond}
	\eeq
	\beq
		\forall B \in \mathscr{B}(\mathcal{Q}): ~~~ \int_0^T dt \int_{\varphi_t(B) \backslash \{ \diamondsuit\}} dq 
\left| J(t,q) \cdot \frac{q}{|q|} \right| < \infty,
	\label{eq:escapecond}
	\eeq
	and, if $\mathcal{Q}=\mathbb{R}^d \backslash \bigcup_{l=1}^m S_l$, in addition that for every $l =1,...,m$:
	\beq
		\exists \delta> 0: \forall B \in \mathscr{B}(\mathcal{Q}): ~~~ \int_0^T dt \int_{\varphi_t(B) \backslash 
\{ \diamondsuit\}} dq ~ \mathbbm{1}(q\in(S_l +\delta)) \frac{|J(t,q) \cdot e_l(q)|}{{\rm dist}(q,S_l)} < \infty.
	\label{eq:critsetcond}
	\eeq
	Here, ${\rm dist}(q,S_l)$ is the Euclidean distance of $q$ from $S_l$ and $e_l(q) = -\nabla_q {\rm dist}(q,S_l)$ 
is the radial unit vector towards $S_l$ at $q \in \mathcal{Q}$.

Then for almost every $q \in \mathcal{Q}$ relative to the measure $\mu_0(dq) = \rho(0,q) dq$, the solution of \eqref{eq:guidance} 
starting at $Q(0) = q$ exists at least up to time $T$, and the family of measures $\mu_t(dq) = \rho(t,q) dq$ is equivariant on $[0,T]$. 
In particular, if \eqref{eq:nodecond}, \eqref{eq:escapecond} and, if appropriate, \eqref{eq:critsetcond} are true for every $T>0$, then 
for $\mu_0$-almost every $q \in \mathcal{Q}$ the solution of \eqref{eq:guidance} starting at $q$ exists for all times $t \geq 0$.
\end{thm}

\paragraph{Remark.} 
The conditions \eqref{eq:nodecond}, \eqref{eq:escapecond} and \eqref{eq:critsetcond} have the following intuitive meanings.
 If \eqref{eq:nodecond} holds, $\mu_0$-almost no trajectory approaches a node for $0 \leq t \leq T$. If \eqref{eq:escapecond} 
is satisfied, $\mu_0$-almost no trajectory escapes to infinity for $0 \leq t \leq T$. If \eqref{eq:critsetcond} holds, $\mu_0$-almost 
no trajectory approaches the critical set $\bigcup_{l=1}^m S_l$ for $0 \leq t \leq T$.

\newpage
\bibliographystyle{plain}

\begin{thebibliography}{[KTZ2020]}

\bibitem[AAV1988]{AAV}
Aharonov, Y., 
Albert, D., 
and 
Vaidman, L.,
``How the Result of Measurement of a Component of the Spin of a Spin-1/2 Particle Can Turn Out to Be 100,''
 \textit{Phys. Rev. Lett.} \textbf{60}:1351--1354 (1988).

\bibitem[AlDy2017]{AlDy}
 Alazzawi, S.,
  and 
Dybalski, W. 
  ``Compton scattering in the Buchholz-Roberts framework of relativistic QED,''
 \textit{Lett. Math. Phys.} \textbf{107}, 81--106 (2017).

\vspace{-3pt}
\bibitem[Betal1995]{BDGPZ1995}
Berndl, K., D\"urr, D., Goldstein, S., Peruzzi, G. and Zangh\`{\i}, N.
``On the global existence of Bohmian mechanics,''
\textit{Commun. Math. Phys.}, {\bf 173}, 647--673 (1995).

\bibitem[Boh1989]{BohmBOOKonQM}
        Bohm, D.,
        \textit{Quantum Theory},
        Dover Pub. (1989).\vspace{-4pt}

\vspace{-3pt}
\bibitem[Boh1952]{Bohm52}
	Bohm, D.,
		``{A suggested interpretation of the quantum theory in terms of ``hidden'' variables. Part I},''
	\textit{Phys. Rev.} \textbf{85}:166--179 (1952);
	         ``Part II,''
	\textit{ibid.}, 180--193 (1952).

\vspace{-3pt}
\bibitem[Bor1926a]{BornsPSISQUAREpapersA}
        Born, M.,
                 ``Zur Quantenmechanik der Stossvorg\"ange,''
        \textit{Z. Phys.} {\textbf{37}},  863--867 (1926).

\vspace{-3pt}
\bibitem[Bor1926b]{BornsPSISQUAREpapersB}
        Born, M.,
                 ``Quantenmechanik der Stossvorg\"ange,''
        \textit{Z. Phys.} {\textbf{38}},  803--827 (1926).


\vspace{-3pt}
\bibitem[deB1927]{deBroglieSOLVAY}
        de Broglie, L.V.P.R.,
        \textit{La nouvelle dynamique des quanta,}
        in ``{Cinqui\`eme Conseil de Physique Solvay}'' (Bruxelles 1927),
        ed. J. Bordet, (Gauthier-Villars, Paris, 1928); 
        English transl.: ``The new dynamics of quanta'', p.374-406 in:
        G. Bacciagaluppi and A. Valentini,
        ``{Quantum Theory at the Crossroads},'' (Cambridge Univ. Press, 2009).  

\vspace{-3pt}
\bibitem[Bri2016]{Bricmont}
  Bricmont, J.
  \emph{Making sense of quantum mechanics}, Springer Verlag (2016).

\vspace{-3pt}
\bibitem[BuRo2014]{Buchholz}
 Buchholz, D., and Roberts, J. E.,
   ``New light on infrared problems: Sectors, statistics, symmetries, and spectrum,''
  \textit{Commun. Math. Phys.} \textbf{330}, 935--972 (2014).

\vspace{-3pt}
\bibitem[CFP2010]{CFPa}
 Chen, T.,
 Fr\"ohlich, J.
 and 
Pizzo, A.
 ``Infraparticle scattering states in non-relativistic QED: I. The Bloch-Nordsieck paradigm.''
  \textit{Commun. Math. Phys.} \textbf{294},  761--825 (2010).

\vspace{-3pt}
\bibitem[CFP2009]{CFPb}
 Chen, T.,
 Fr\"ohlich, J.
 and 
 Pizzo, A.
 ``Infraparticle scattering states in non-relativistic QED: II. Mass shell properties.''
  \textit{J. Math. Phys.} \textbf{50}, 012103--012134 (2009).

\vspace{-3pt}
\bibitem[DeNi2016]{DN2016}
Deckert, D.-A. and Nickel, L.,
``Consistency of multi-time Dirac equations with general interaction potentials,''
\textit{J. Math. Phys.}, {\bf 57}(7):072301 (2016).


\vspace{-3pt}
\bibitem[DeNi2019]{DN2019}
Deckert, D.-A. and Nickel, L.,
``Multi-time dynamics of the Dirac-Fock-Podolsky model of QED,''
\textit{J. Math. Phys.}, {\bf 60}: 072301 (2019).

\vspace{-3pt}
\bibitem[Dir1932]{Dir1932}
Dirac, P. A. M.,
``Relativistic Quantum Mechanics,''
\textit{Proc. R. Soc. Lond. A}, {\bf 136}:453--464 (1932).

\vspace{-3pt}
\bibitem[Detal1999]{DGMZ1999}
D\"urr, D., Goldstein, S., M\"unch-Berndl, K., and Zangh\`{\i}, N.,
``Hypersurface Bohm--Dirac models''
\textit{Phys. Rev. A}, {\bf 60}: 2729--2736 (1999).

\vspace{-3pt}
\bibitem[D\"uTe2009]{DT2009}
D\"urr, D. and Teufel, S.,
\textit{Bohmian Mechanics},
Springer, 2009.

\vspace{-3pt}
\bibitem[Detal2014]{DGNSZ2014}
D\"urr, D., Goldstein, S., Norsen, T., Struyve, W. and Zangh\`{\i}, N.,
``Can Bohmian mechanics be made relativistic?''
\textit{Proc. R. Soc. A}, {\bf 470}(2162):20130699 (2014).

\vspace{-3pt}
\bibitem[Ein1909a]{EinsteinPHOTONb}
  Einstein, A.,
 ``Zum gegenw\"artigen Stand des Strahlungsproblems,''
 \textit{Phys. Zeitschr.} \textbf{10}, 185--193 (1909). 

\vspace{-3pt}
\bibitem[Ein1909b]{EinsteinPHOTONbb}
  Einstein, A.,
  ``\"Uber die Entwicklung unserer Anschauungen \"uber das Wesen und die Konstitution der Strahlung,''
 \textit{Verh. Deutsch. Phys. Ges.} \textbf{7}, 482--500 (1909). 

\vspace{-3pt}
\bibitem[FGS2004]{FGS}
 Fr\"ohlich, J.,
 Griesemer, M., 
 and
 Schlein, B.,
  ``Asymptotic completeness for Compton scattering,''
 \textit{Commun. Math. Phys.} \textbf{252}, 415--476 (2004). 

\vspace{-3pt}
\bibitem[Gar1998]{Gar1998}
Garabedian, P. R.,
\textit{Partial Differential Equations},
AMS Chelsea Pub., 672 pp. (1998).

\vspace{-3pt}
\bibitem[Gue1995]{Guerra1995}
Guerra, F.,
\emph{Introduction to Nelson's Stochastic Mechanics as a Model for Quantum Mechanics},
pp. 339--355 in ``The Foundations of Quantum Mechanics,'' C. Garola and A. Rossi (eds.),
Kluwer, Amsterdam (1995).

\vspace{-3pt}
\bibitem[Jos2002]{Jost}
 Jost, R.,
\emph{Das M\"archen vom elfenbeinernen Turm},
 Springer, Wien (2002).

\vspace{-3pt}
\bibitem[KTZ2018]{KTZ2018}
Kiessling, M. K.-H., and Tahvildar-Zadeh, A. S.,
``On the quantum mechanics of a single photon,''
\textit{J. Math. Phys.}, {\bf 59}:112302 (2018).

\vspace{-3pt}
\bibitem[Ketal2011]{KocsisETal}
Kocsis, S.,
 Braverman, B.,
 Ravets, S., 
Stevens, M. J.,
Mirin, R. P.,
Shalm,  L. K., 
and
Steinberg,  A. M.,
``Observing the Average Trajectories of Single Photons in a Two-Slit Interferometer,''
\emph{Science} \textbf{332}:1170--1173 (2011); DOI: 10.1126/science.1202218

\vspace{-3pt}
\bibitem[Lie2015]{Lie2015}
Lienert, M.,
``A relativistically interacting exactly solvable multi-time model for two massless Dirac particles in 1+1 dimensions,''
\textit{J. Math. Phys.} {\bf 56}:042301 (2015).

\vspace{-3pt}
\bibitem[LiNi2015]{LN2015}
Lienert, M. and Nickel, L.,
``A simple explicitly solvable interacting relativistic $N$-particle model,''
\textit{J. Phys. A: Math. Theor.}, {\bf 48}(32), 325301 (2015).


\vspace{-3pt}
\bibitem[LiNi2020]{LN2020}
Lienert, M. and Nickel, L.,
``Multi-time formulation of particle creation and annihilation via interior-boundary conditions,''
{\em Rev. Math. Phys.}, {\bf 32}, 2050004 (2020).

\vspace{-3pt}
\bibitem[LPT2017]{LPT2017}
{Lienert}, M.,  {Petrat}, S., and {Tumulka}, R.,
``Multi-time wave functions,''
{\em J. Phys. Conf. Ser.}, {\bf 880}(1):012006 (2017).


\vspace{-3pt}
\bibitem[LiTu2017]{LT2017}
Lienert, M. and Tumulka, R.,
``Born's rule for arbitrary Cauchy surfaces,'' 
{\em preprint}, (2017). {\url{https://arxiv.org/abs/1706.07074}}. 
% Without the arXiv identifier or link, there is no complete reference in this case.

\vspace{-3pt}
\bibitem[Pen1990]{Penrose}
Penrose, R.,
 \emph{The emperor's new mind}, 2nd ed., Oxford Univ. Press, New York (1990).

\vspace{-3pt}
\bibitem[PeTu2014a]{PT2014a}
Petrat, S. and Tumulka, R., 
``Multi-time Schr\"odinger equations cannot contain interaction potentials,''
{\em J. Math. Phys.}, {\bf 55}, 032302 (2014).


\vspace{-3pt}
\bibitem[PeTu2014b]{PT2014b}
Petrat, S. and Tumulka, R.,
``Multi-time wave functions for quantum field theory,''
{\em Ann. Phys.}, {\bf 345}:17--54 (2014).


\vspace{-3pt}
\bibitem[Rie1946]{Rie1946}
Riesz, M., 
{\it Sur certaines notions fondamentales en th\'eorie quantique relativiste,}
 In ``Dixieme Congres Math. des Pays Scandinaves," 123--148, Copenhagen (1946).

\vspace{-3pt}
\bibitem[Schw1994]{SchweberBOOK}
Schweber, S.,
{\it QED and the men who made it: Dyson, Feynman, Schwinger, and Tomonaga},
 Princeton Univ. Press (1994).

\vspace{-3pt}
\bibitem[TeTu2005]{TT2005}
  Teufel, S. and Tumulka, R.,
  ``Simple Proof for Global Existence of Bohmian Trajectories,''
{\em Commun. Math. Phys.}, {\bf 258}:349--365 (2005).

\vspace{-3pt}
\bibitem[TeTu2015]{TT2015}
  Teufel, S. and Tumulka, R.,
   ``Hamiltonians without ultraviolet divergence for quantum field theories,  https://arxiv.org/abs/1505.04847  (2015).

\vspace{-3pt}
\bibitem[TeTu2016]{TT2016}
  Teufel, S. and Tumulka, R.,
   ``Avoiding ultraviolet divergence by means of interior–boundary conditions,'' 
  in: Finster F., Kleiner J., R\"oken C., Tolksdorf J. (eds.) \emph{Quantum Mathematical Physics}.
  Birkhäuser, Cham (2016).

\vspace{-3pt}
\bibitem[Tha1992]{ThallerBOOK}
  Thaller, B.,
  {\em The Dirac Equation},
  Springer Verlag, Berlin (1992).

\bibitem[Wei1995]{WeinbergBOOKqft}
  Weinberg, S.
  ``The quantum theory of fields, Vol. I,''
  Cambridge Univ. Press, (1995).
\end{thebibliography}

\end{document}